\pdfoutput=1
\documentclass[a4paper,UKenglish,cleveref, autoref, thm-restate]{lipics-v2021}
\hideLIPIcs
\graphicspath{{./images/}}

\title{Subtyping Context-Free Session Types} 

\author{Gil Silva}{LASIGE, Faculdade de Ciências da Universidade de Lisboa, Lisbon, Portugal}{gsilva@lasige.di.fc.ul.pt}{https://orcid.org/0009-0007-7051-6116}{}

\author{Andreia Mordido}{LASIGE, Faculdade de Ciências da Universidade de Lisboa, Lisbon, Portugal \and \url{http://www.di.fc.ul.pt/~amordido/}}{afmordido@ciencias.ulisboa.pt}{https://orcid.org/0000-0002-1547-0692}{}

\author{Vasco T. Vasconcelos}{LASIGE, Faculdade de Ciências da Universidade de Lisboa, Lisbon, Portugal \and \url{https://www.di.fc.ul.pt/~vv/}}{vmvasconcelos@ciencias.ulisboa.pt}{https://orcid.org/0000-0002-9539-8861}{}

\authorrunning{G. Silva, A. Mordido and V. T. Vasconcelos} 

\Copyright{Gil Silva, Andreia Mordido and Vasco T. Vasconcelos} 



\category{}

\relatedversion {}
\relatedversiondetails[cite=DBLP:conf/concur/SilvaMV23]{In Proceedings of the 34th International Conference on Concurrency Theory (CONCUR 2023)}
{https://doi.org/10.4230/LIPIcs.CONCUR.2023.11}

\funding{Support for this research was provided by the Fundação para a Ciência e a Tecnologia through project SafeSessions, ref.\ PTDC/CCI-COM/6453/2020, and the LASIGE Research Unit, ref.\ UIDB/00408/2020 and ref.\ UIDP/00408/2020.}

\acknowledgements{We thank Luca Padovani, Diana Costa, Diogo Poças and the anonymous reviewers for their insightful comments.}

\nolinenumbers

\EventEditors{Guillermo A. P\'{e}rez and Jean-Fran\c{c}ois Raskin}
\EventNoEds{2}
\EventLongTitle{34th International Conference on Concurrency Theory (CONCUR 2023)}
\EventShortTitle{CONCUR 2023}
\EventAcronym{CONCUR}
\EventYear{2023}
\EventDate{September 18--23, 2023}
\EventLocation{Antwerp, Belgium}
\EventLogo{}
\SeriesVolume{279}
\ArticleNo{11}

\usepackage{mathpartir}
\usepackage{stmaryrd}
\usepackage{scalerel}
\usepackage{trimclip}
\usepackage{nccmath}
\usepackage{mathtools}
\usepackage{cancel}

\newcommand{\tUnit}{\mathsf{Unit}}
\newcommand{\tInt}{\mathsf{Int}}
\newcommand{\tBool}{\mathsf{Bool}}

\newcommand{\lin}{\mathsf{1}}
\newcommand{\un}{\mathbf{\ast}}
\newcommand{\choicebranch}{\odot}
\newcommand{\inout}{\sharp}
\newcommand{\tEnd}{\mathsf{End}}
\newcommand{\tSkip}{\mathsf{Skip}}
\newcommand{\tArrow}[3]{#1\overset{#2}{\rightarrow}#3}
\newcommand{\tMsg}[1]{{\inout}#1}
\newcommand{\tIn}[1]{{{?}#1}}
\newcommand{\tOut}[1]{{{!}#1}}
\newcommand{\tChoice}[3]{{\odot}\mathtt{\{} #1 {:}\, #2 \mathtt{\}}_{#1\in #3}}
\newcommand{\tIntChoice}[3]{\oplus \mathtt{\{} #1 {:}\, #2 \mathtt{\}}_{#1\in #3}}
\newcommand{\tExtChoice}[3]{\& \mathtt{\{} #1 {:}\, #2 \mathtt{\}}_{#1\in #3}}
\newcommand{\tSeq}[2]{{#1}{;}{#2}}
\newcommand{\tRecord}[3]{\mathtt{\{} #1 {:}\, #2 \mathtt{\}}_{#1\in #3}}
\newcommand{\tRcdVrt}[3]{\mathtt{\llparenthesis} #1 {:}\, #2 \mathtt{\rrparenthesis}_{#1\in #3}}
\newcommand{\tVariant}[3]{\mathtt{\langle} #1 {:}\, #2 \mathtt{\rangle}_{#1\in #3}}
\newcommand{\subst}[3]{[#2/#1] #3}

\newcommand{\tRec}[2]{\mu{#1}{.}#2}

\newcommand{\field}[2]{\mathsf{#1}{:}\,#2}

\newcommand{\lChoiceDefault}{\odot}
\newcommand{\lChoiceLabel}[1]{\odot_\mathsf{#1}}
\newcommand{\lIntChoiceDefault}{\oplus}
\newcommand{\lIntChoiceLabel}[1]{\oplus_{#1}}

\newcommand{\lExtChoiceLabel}[1]{\&_{#1}}
\newcommand{\lMsgPayload}{\sharp\mathsf{p}}
\newcommand{\lMsgContinuation}{\sharp\mathsf{c}}
\newcommand{\lOutPayload}{!\mathsf{p}}
\newcommand{\lOutContinuation}{!\mathsf{c}}
\newcommand{\lInPayload}{?\mathsf{p}}
\newcommand{\lInContinuation}{?\mathsf{c}}
\makeatletter
\DeclareRobustCommand{\shortto}{%
	\mathpalette\short@to\relax%
}

\newcommand{\short@to}[2]{%
	\mkern2mu
	\clipbox{{.5\width} 0 0 0}{$\m@th#1\vphantom{+}{\to}$}%
}
\makeatother
\newcommand{\lArrowDomain}{{\shortto\mathsf{d}}}
\newcommand{\lArrowRange}{{\shortto\mathsf{r}}}
\newcommand{\lArrowLinear}{\shortto\lin}
\newcommand{\transition}[3]{#1 \overset{#2}{\longrightarrow} #3}
\newcommand{\transitionGNF}[4]{#2 \overset{#3}{\longrightarrow}_#1 #4}

\newcommand{\svert}{\enskip{|}\enskip}
\newcommand{\terminated}[1]{#1\checkmark}
\newcommand{\notTerminated}[1]{{#1\bcancel\checkmark}}
\newcommand{\lRcdVrtDefault}{\llparenthesis\rrparenthesis}
\newcommand{\lRcdVrtLabel}[1]{\llparenthesis\rrparenthesis_\mathsf{#1}}
\newcommand{\lRcdDefault}{\{\}}
\newcommand{\lRcdLabel}[1]{\{\}_{#1}}
\newcommand{\lVrtDefault}{\langle\rangle}
\newcommand{\lVrtLabel}[1]{\langle\rangle_{#1}}

\newcommand{\subS}{\leq}
\newcommand{\simS}{\lesssim}

\newcommand{\simSX}{\mathcal{X}}
\newcommand{\simSY}{\mathcal{Y}}
\newcommand{\simSZ}{\mathcal{Z}}
\newcommand{\simSW}{\mathcal{W}}
\newcommand{\xyzwsim}{\preceq^{\simSX\simSY\simSZ\simSW}}
\newcommand{\prods}{\mathcal{P}}
\newcommand{\xyzwsimP}{\xyzwsim_\prods}
\newcommand{\wellFormed}[2]{#1 \vdash #2}
\newcommand{\contr}[2]{#1 \mathrel{\mathsf{contr}} #2}

\newcommand{\fun}[1]{\mathit{#1}}
\newcommand{\singletonQueue}{\fun{singletonQueue}}
\newcommand{\emptyf}{\fun{empty}}
\newcommand{\front}{\fun{front}}

\newcommand{\dequeue}{\fun{dequeue}}
\newcommand{\simplify}{\fun{simplify}}

\newcommand{\word}{\fun{word}}
\newcommand{\grm}{\fun{grm}}
\newcommand{\subG}{\fun{subG}}
\newcommand{\subT}{\fun{subT}}
\newcommand{\prune}{\fun{prune}}
\newcommand{\unravel}{\fun{unr}}
\newcommand{\unrOne}{\fun{unr1}}

\newcommand{\free}{\fun{free}}

\newcommand{\True}{\textnormal{\textbf{True}}}
\newcommand{\False}{\textnormal{\textbf{False}}}

\begin{document}
\maketitle

\begin{abstract}
Context-free session types describe structured patterns of communication on heterogeneously typed channels, allowing the specification of protocols unconstrained by tail recursion. The enhanced expressive power provided by non-regular recursion comes, however, at the cost of the decidability of subtyping, even if equivalence is still decidable. We present an approach to subtyping context-free session types based on a novel kind of observational preorder we call $\mathcal{XYZW}$-simulation, which generalizes $\mathcal{XY}$-simulation (also known as covariant-contravariant simulation) and therefore also bisimulation and plain simulation. We further propose a subtyping algorithm that we prove to be sound, and present an empirical evaluation in the context of a compiler for a programming language. Due to the general nature of the simulation relation upon which it is built, this algorithm may also find applications in other domains.
\end{abstract}

\section{Introduction}\label{sec:introduction}
Session types, introduced by Honda et al.~\cite{DBLP:conf/concur/Honda93, DBLP:conf/esop/HondaVK98, DBLP:conf/parle/TakeuchiHK94}, enhance traditional type systems with the ability to specify and enforce structured communication protocols on bidirectional, heterogeneously typed channels. Typically, these specifications include the type, direction (input or output) and order of the messages, as well as branching points where one participant can choose how to proceed and the other must follow.

Traditional session types are bound by tail recursion and therefore restricted
to the specification of protocols described by regular languages. This excludes
many protocols of practical interest, with the quintessential example being the
serialization of tree-structured data on a single channel. Context-free session
types, proposed by Thiemann and Vasconcelos~\cite{DBLP:conf/icfp/ThiemannV16},
liberate types from tail recursion by introducing a sequential composition operator ($\tSeq{\_}{\_}$) with a monoidal structure and a left and right identity in type $\tSkip$, representing no action. As their name hints, context-free session types can specify protocols corresponding to (simple deterministic) context-free languages and are thus considerably more expressive than their regular counterparts.

What does it mean for a context-free session type to be a subtype of another? Our answer follows Gay and Hole's seminal work on subtyping for regular session types~\cite{DBLP:journals/acta/GayH05}, and Liskov's \emph{principle of safe substitution}~\cite{DBLP:conf/oopsla/Liskov87}: $S$ is a subtype of $R$ if channels governed by type $S$ can take the place of channels governed by type $R$ in whatever context, without violating the guarantees offered by a type system (e.g.\ progress, deadlock freedom, session fidelity, etc.).

More concretely, subtyping allows increased flexibility in the interactions between participants, namely on the type of the messages (a feature inherited from the subtyped $\pi$-calculus~\cite{DBLP:journals/mscs/PierceS96}) and on the choices available at branching points~\cite{DBLP:journals/acta/GayH05}, allowing a channel to be governed by a simpler session type if its context so requires. A practical benefit of this flexibility is that it promotes \emph{modular development}: the behaviour of one participant may be refined, while the behaviour of the other is kept intact.
\begin{example}\label{ex:intro} Consider the following context-free session types for serializing binary trees.
	\begin{align*}
		\begin{aligned}[t]
			\mathsf{STree} &= \tRec{s}{{\oplus}\{\field{Nil}{\tSkip}, \field{Node}{\tSeq{s}{\tSeq{\tOut{\tInt}}{s}}}\}}\\
			\mathsf{DTree} &= \tRec{s}{\&\{\field{Nil}{\tSkip}, \field{Node}{\tSeq{s}{\tSeq{\tIn{\tInt}}{s}}}\}}\\
		\end{aligned}
		\enskip 
		\begin{aligned}[t]
			\mathsf{SEmpty} &= \oplus\{\field{Nil}{\tSkip}\}\\
			\mathsf{SFullTree0} &= \oplus\{\field{Node}{\tSeq{\mathsf{SEmpty}}{\tSeq{\tOut\tInt}{\mathsf{SEmpty}}}}\}\\
			\mathsf{SFullTree1} &= \oplus\{\field{Node}{\tSeq{\tSeq{\mathsf{SFullTree0}}{\tOut\tInt}}{\mathsf{SFullTree0}}}\}
		\end{aligned}
	\end{align*}
	The recursive $\mathsf{STree}$ and $\mathsf{DTree}$ types specify, respectively,
	the serialization and deserialization of a possibly infinite arbitrary tree,
	while the remaining non-recursive types specify the serialization of finite
	trees of particular configurations. The benefit of subtyping is that it makes
	the particular types $\mathsf{SEmpty}$, $\mathsf{SFullTree0}$ and
	$\mathsf{SFullTree1}$ compatible with the general $\mathsf{DTree}$ type. Observe
	that its dual, $\mathsf{STree}$, may safely take the place of any type in the
	right column. Consider now a function $\mathsf{f}$ that generates full trees of
	height 1 and serializes them on a given channel end. Assigning it type
	$\tArrow{\mathsf{STree}}{}{\tUnit}$ would not statically ensure that the
	fullness and height of the tree are as specified. Type
	$\tArrow{\mathsf{SFullTree1}}{}{\tUnit}$ would do so, and subtyping would still
	allow the function to use an $\mathsf{STree}$ channel (i.e., communicate with
	someone expecting an arbitrary $\mathsf{DTree}$ tree).
\end{example}
Expressive power usually comes at the cost of decidability. While subtyping for regular session types has been formalized, shown decidable and given an algorithm by Gay and Hole~\cite{DBLP:journals/acta/GayH05}, subtyping in the context-free setting has been proven undecidable by Padovani~\cite{DBLP:journals/toplas/Padovani19}. The proof is given by a reduction from the inclusion problem for simple languages, shown undecidable by Friedman~\cite{DBLP:journals/tcs/Friedman76}. Remarkably, the equivalence problem for simple languages is known to be decidable, as is the type equivalence of context-free session types~\cite{DBLP:conf/focs/KorenjakH66, DBLP:conf/icfp/ThiemannV16}.

Subtyping context-free session types has until now been considered only in a limited form, where message types must be syntactically equal~\cite{DBLP:journals/toplas/Padovani19}. Consequently, the interesting co/contravariant properties of input/output types have been left unexplored. In this paper, we propose a more expressive subtyping relation, where the types of messages may vary co/contravariantly, according to the classical subtyping notion of Gay and Hole. To handle the contravariance of output types, we introduce a novel notion of observational preorder, which we call \emph{$\mathcal{XYZW}$-simulation} (by analogy with \emph{$\mathcal{XY}$-simulation}~\cite{DBLP:conf/concur/AartsV10}).

While initially formulated in the context of the $\pi$-calculus, considerable
work has been done to integrate session types in more standard settings, such as
functional languages based on the polymorphic $\lambda$-calculus with linear
types~\cite{DBLP:journals/iandc/AlmeidaMTV22, DBLP:journals/corr/abs-2203-12877,
	DBLP:conf/esop/PocasCMV23}. In this scenario, functional
types and session types are not orthogonal: sessions may carry functions, and functions may act on sessions. With this in mind,
we promote our theory to a linear functional setting, thereby showing how
subtyping for records, variants and (linear and unrestricted~\cite{C:techreport/Gay06}) functions, usually introduced by inference rules, can be seamlessly integrated with simulation-based subtyping for
context-free session types.

Finally, we present a sound algorithm for the novel notion of subtyping, based on the type equivalence algorithm of Almeida et al.~\cite{DBLP:conf/tacas/AlmeidaMV20}. This algorithm works by first encoding the types as words in a simple grammar~\cite{DBLP:conf/focs/KorenjakH66} and then deciding their $\mathcal{XYZW}$-similarity. Being grammar-based and, at its core, agnostic to types, our algorithm may also find applications for other objects with similar non-regular and contravariant properties.
\subparagraph*{Contributions} We address the subtyping problem for context-free session types, proposing:
\begin{itemize}
	\item A syntactic definition of subtyping for context-free session types;
	\item A novel kind of behavioural preorder called $\mathcal{XYZW}$-simulation, and, based on it, a semantic definition of subtyping that coincides with the syntactic one;
	\item A sound subtyping algorithm based on the $\mathcal{XYZW}$-similarity of simple grammars;
	\item An empirical evaluation of the performance of the algorithm, and a comparison with an existing type equivalence algorithm.	
\end{itemize}
\subparagraph*{Overview} The rest of this paper is organized as follows: in
\cref{sec:cfst} we introduce types, type formation and
syntactic subtyping; in \cref{sec:semantic-subtyping} we present a notion
of semantic subtyping, to be used as a stepping stone to develop our subtyping
algorithm; in \cref{sec:algorithm} we present the algorithm and show it
to be sound with respect to the semantic subtyping relation; in \cref{sec:evaluation} we evaluate the performance of our
implementation of the algorithm; in \cref{sec:related-work} we present
related work; in \cref{sec:conclusion} we conclude the paper and trace a
path for the work to follow. The reader can find the rules for type formation
and proofs for all results in the paper in the appendices.

\section{Types and syntactic subtyping}
\label{sec:cfst}
We base our contributions on a type language that includes both functional types
and higher-order context-free session types (i.e., types that allow
messages of arbitrary types). The language is shown in \cref{fig:typesyntax}. As customary in session types for functional
languages~\cite{DBLP:journals/jfp/GayV10}, the language of types is given by two
mutually recursive syntactic categories: one for functional types
and another for session types. We assume two disjoint and denumerable
sets of type references, with the first ranged over by $t,u,v,w$, the second by
$r,s$ and their union by $x,y,z$. We further assume a set of record, variant and
choice labels, ranged over by $j,k,\ell$.
\begin{figure}[t]
	\text{Functional and higher-order context-free session types}
	\begin{ceqn}
		\begin{align*}
			T,U,V,W &\Coloneqq \tUnit 
			\svert \tArrow{T}{m}{U} 
			\svert \tRcdVrt{\ell}{T}{L} 
			\svert S 
			\svert t
			\svert \tRec{t}{T}\\
			S,R &\Coloneqq \tMsg{T} 
			\svert \tChoice{\ell}{T}{L}  
			\svert \tSkip
			\svert \tEnd
			\svert \tSeq{S}{R}
			\svert s
			\svert \tRec{s}{S}
		\end{align*}
	\end{ceqn}
	\text{Multiplicities, records/variants, polarities and views}
	\begin{ceqn}
		\begin{gather*}
			m,n \Coloneqq {\lin} \svert \mathbf{*} \qquad 
			\llparenthesis \cdot \rrparenthesis \Coloneqq \{\cdot\} \svert \langle \cdot \rangle \qquad \inout \Coloneqq {?} \svert {!} \qquad  \choicebranch \Coloneqq \oplus \svert \&
		\end{gather*}
		\caption{Syntax of types.}
		\label{fig:typesyntax}
	\end{ceqn}
\end{figure}

The first three productions of the grammar for functional types introduce the
$\tUnit$ type, functions $\tArrow{T}{m}{U}$, records $\tRecord{\ell}{T_\ell}{L}$
and variants $\tVariant{\ell}{T_\ell}{L}$ (which correspond to
\textit{datatypes} in ML-like languages). Our system exhibits linear
characteristics: function types contain a multiplicity annotation $m$ (also in
\cref{fig:typesyntax}), meaning that they must be used exactly once if
$m={\lin}$ or without restrictions if $m=\mathbf{\ast}$ (such types can also be
found, for instance, in Gay's proposal~\cite{DBLP:journals/jfp/GayV10}, in System
F$^\circ$~\cite{DBLP:conf/tldi/MazurakZZ10} and in the FreeST
language~\cite{DBLP:journals/iandc/AlmeidaMTV22}). Their inclusion in our system
is justified by the interesting subtyping properties they exhibit~\cite{C:techreport/Gay06}.

Session types $\tOut{T}$ and $\tIn{T}$ represent the sending and receiving,
respectively, of a value of type $T$ (an arbitrary type, making the system
higher-order). Internal choice types $\tIntChoice{\ell}{S_\ell}{L}$ allow the
selection of a label $k\in L$ and its continuation $S_k$, while external
choice types $\tExtChoice{\ell}{S_\ell}{L}$ represent the branching on any
label $k\in L$ and its continuation $S_k$. We stipulate that the set of labels
for these types must be non empty. Type $\tSkip$ represents no action, while
type $\tEnd$ indicates the closing of a channel, after which no more
communication can take place. Type $\tSeq{R}{S}$ denotes the sequential
composition of $R$ and $S$, which is associative, right distributes over choices
types, has (left and right) identity $\tSkip$ and left-absorber $\tEnd$.

The final two productions in both functional and session grammars introduce self-references and the recursion operator. Their inclusion in the two grammars ensures we can have both recursive functional types and recursive session types while avoiding nonsensical types such as $\tRec{t}{\tArrow{\tUnit}{*}{\tSeq{\tOut{\tUnit}}{t}}}$ at the syntactical level (avoiding the need for a kinding system). 

Still, we do not consider all types generated by these grammars to be \emph{well-formed}. Consider session type $\tRec{r}{\tSeq{r}{\tOut{\tUnit}}}$. No matter how many times we unfold it, we cannot resolve its first communication action. The same could be said of $\tRec{r}{\tSeq{\tSkip}{\tSeq{r}{\tOut{\tUnit}}}}$. We must therefore ensure that any self-reference in a sequential composition is preceded by a type constructor representing some meaningful action, i.e., not equivalent to $\tSkip$. This is achieved by adapting the conventional notion of contractivity (no subterms of the form $\tRec{x}{\tRec{x_1}{\,\ldots\,\tRec{x_n}{x}}}$)~\cite{DBLP:journals/acta/GayH05} to account for $\tSkip$ as the identity of sequential composition. This corresponds to the notion of \textit{guardedness} in the theory of process algebra (e.g.~\cite{ DBLP:journals/iandc/GrooteH94,DBLP:books/sp/Milner80}).

In addition to contractivity, we must ensure that well-formed types contain no
free references. The type formation judgement $\wellFormed{\Delta}{T}$, where
$\Delta$ is a set of references, combines these requirements. The rules for the
judgement can be found in \cref{subsec:type-formation}.
\begin{figure}[t!]
	\begin{mathpar}
		\text{Syntactic subtyping (\emph{coinductive})}\hfill\fbox{$T \subS T$}\\
		\inferrule[S-Unit]{}{\tUnit \subS \tUnit}
		
		\inferrule[S-Arrow]{U_1 \subS T_1 \\ T_2 \subS U_2 \\ m \sqsubseteq n}{T_1 \overset{m}\to T_2 \subS U_1 \overset{n}\to U_2}
		
		\inferrule[S-Rcd]
		{K \subseteq L \\ T_j \subS U_j \,\,(\forall j\in K)}
		{\{\ell{:}\,T_\ell\}_{\ell \in L} \subS \{k{:}\,U_k\}_{k \in K}}
		
		\inferrule[S-Vrt]
		{L \subseteq K \\ T_j \subS U_j \,\,(\forall j\in L)}
		{\langle \ell{:}\,T_\ell\rangle_{\ell \in L} \subS \langle k{:}\,U_k\rangle_{k \in K}}
		
		\inferrule[S-RecL]
		{\subst{x}{\tRec{x}{T}}{T}\subS U}
		{\tRec{x}{T} \subS U}
		
		\inferrule[S-RecR]
		{T \subS \subst{x}{\tRec{x}{U}}{U}}
		{T \subS \tRec{x}{U}}
		
		\inferrule[S-In]{T \subS U}{\tIn T \subS \tIn U}
		
		\inferrule[S-Out]{U \subS T}{\tOut T \subS \tOut U}
		
		\inferrule[S-ExtChoice]
		{L \subseteq K \\ S_j \subS R_j \,\,(\forall j\in L)}
		{\&\{ \ell{:}\,S_\ell\}_{\ell \in L} \subS \&\{k{:}\,R_k\}_{k \in K}}
		
		\inferrule[S-IntChoice]
		{K \subseteq L \\ S_j \subS R_j \,\,(\forall j\in K)}
		{\oplus\{\ell{:}\,S_\ell\}_{\ell \in L} \subS \oplus\{k{:}\,R_k\}_{k \in K}}
		
		\inferrule[S-Skip]{}{\tSkip \subS \tSkip}
		
		\inferrule[S-End]{}{\tEnd \subS \tEnd}
		
		\inferrule[S-InSeq1L]{T \subS U \\ S\subS \tSkip}{\tSeq{\tIn T}{S} \subS \tIn U}
		
		\inferrule[S-InSeq1R]{T \subS U \\ S\subS\tSkip}{\tIn T \subS \tSeq{\tIn{U}}{S}}
		
		\inferrule[S-InSeq2]{T \subS U \\ S\subS R}{\tSeq{\tIn T}{S} \subS \tSeq{\tIn{U}}{R}}
		
		\inferrule[S-OutSeq1L]{U \subS T \\ S\subS\tSkip}{\tSeq{\tOut{T}}{S} \subS \tOut{U}}
		
		\inferrule[S-OutSeq1R]{U \subS T \\ S\subS\tSkip}{\tOut T \subS \tSeq{\tOut{U}}{S}}
		
		\inferrule[S-OutSeq2]{U \subS T \\ S\subS R}{\tSeq{\tOut T}{S} \subS \tSeq{\tOut U}{R}}
		
		\inferrule[S-ChoiceSeqL]
		{\tChoice{\ell}{\tSeq{S_\ell}{S}}{L} \subS R}
		{\tSeq{\tChoice{\ell}{S_\ell}{L}}{S} \subS R}
		
		\inferrule[S-ChoiceSeqR]
		{S \subS \tChoice{\ell}{\tSeq{R_\ell}{R}}{L}}
		{S \subS \tSeq{\tChoice{\ell}{R_\ell}{L}}{R}}
		
		\inferrule[S-SkipSeqL]{S \subS R}{\tSeq{\tSkip}{S} \subS R}
		
		\inferrule[S-SkipSeqR]{S \subS R}{S \subS \tSeq{\tSkip}{R}}
		
		\inferrule[S-EndSeq1L]{}{\tSeq \tEnd S \subS \tEnd}
		
		\inferrule[S-EndSeq1R]{}{\tEnd \subS \tSeq \tEnd R}
		
		\inferrule[S-EndSeq2]{}{\tSeq{\tEnd}{S} \subS \tSeq{\tEnd}{R}}
		
		\inferrule[S-SeqSeqL]
		{\tSeq{S_1}{(\tSeq{S_2}{S_3})} \subS R}
		{\tSeq{(\tSeq{S_1}{S_2})}{S_3} \subS R}
		
		\inferrule[S-SeqSeqR]
		{S \subS \tSeq{R_1}{(\tSeq{R_2}{R_3})}}
		{S \subS \tSeq{(\tSeq{R_1}{R_2})}{R_3}}
		
		\inferrule[S-RecSeqL]
		{\tSeq{(\subst{s}{\tRec{s}{S_1}}{S_1})}{S_2} \subS R}
		{\tSeq{(\tRec{s}{S_1})}{S_2} \subS R}
		
		\inferrule[S-RecSeqR]
		{S \subS \tSeq{(\subst{s}{\tRec{s}{R_1}}{R_1})}{R_2}}
		{S \subS \tSeq{(\tRec{s}{R_1})}{R_2}}\\
		\text{Preorder on multiplicities}\hfill\fbox{$m\sqsubseteq m$}\\
		m \sqsubseteq m \hspace{1.5em} \ast \sqsubseteq {\lin}
	\end{mathpar}
	\caption{Syntactic subtyping.}
	\label{fig:syntactic-sub}
\end{figure}
We are now set to define our syntactic subtyping relation. We begin by surveying the features it should support:
\begin{itemize}
	\item\textbf{Input and output subtyping} Input variance and output contravariance are the central features of subtyping for types that govern entities that can be written to or read from, such as channels and references~\cite{DBLP:books/daglib/0005958}. They are therefore natural features of the subtyping relation for conventional session types as well~\cite{DBLP:journals/acta/GayH05}. Observe that~$\tIn{\{\field{A}{\tInt},\field{B}{\tBool}\}}\subS\tIn{\{\field{A}{\tInt}\}}$ should be true, for the type of the received value, $\{\field{A}{\tInt},\field{B}{\tBool}\}$, safely substitutes the expected type, $\{\field{A}{\tInt}\}$. Observe also that $\tOut{\{\field{A}{\tInt}\}}\subS\tOut{\{\field{A}{\tInt},\field{B}{\tBool}\}}$ should be true, because the type of the value to be sent, $\{\field{A}{\tInt},\field{B}{\tBool}\}$, is a subtype of $\{\field{A}{\tInt}\}$, the type of the messages the substitute channel is allowed to send.
	
	\item \textbf{Choice subtyping} If we understand external and internal choice types as, respectively, the input and output of a label, then their subtyping properties are easy to derive: external choices are covariant on their label set, internal choices are contravariant on their label set, and both are covariant on the continuation of the labels (this is known as \textit{width subtyping}). Observe that $\&\{\field{A}{\tIn{\tInt}}\} \subS \&\{\field{A}{\tIn{\tInt}},\field{B}{\tOut{\tBool}}\}$ should be true, for every branch in the first type can be safely handled by matching on the second type. Likewise, $\oplus\{\field{A}{\tIn{\tInt}},\field{B}{\tOut{\tBool}}\} \subS \oplus\{\field{A}{\tIn{\tInt}}\}$ should be true, for every choice in the second type can be safely selected in the first.
	
	\item\textbf{Sequential composition} In the classical subtyping relation for regular session types, input and output types ($\tMsg{T}.S$) can be characterized as covariant in their continuation. Although the same general intuition applies in the context-free setting, we cannot as easily characterize the variance of the sequential composition constructor ($\tSeq{S}{R}$) due to its monoidal, distributive and absorbing properties. For instance, consider types $\tSeq{S_1}{S_2}$ and $\tSeq{R_1}{R_2}$, with $S_1=\tSeq{\tOut{\tInt}}{\tOut{\tBool}}$, $S_2=\tIn{\tInt}, R_1=\tOut{\tInt}$ and $R_2=\tSeq{\tOut{\tBool}}{\tIn{\tInt}}$. Although it should be true that $\tSeq{S_1}{S_2}\subS \tSeq{R_1}{R_2}$, we can have neither $S_1 \subS R_1$ nor $S_2 \subS R_2$.
	
	\item\textbf{Functional subtyping} The subtyping properties of function, record and variant types are well known, and we refer the readers to Pierce's book for the reasoning behind them~\cite{DBLP:books/daglib/0005958}. Succinctly, the function type constructor is contravariant on the domain and covariant on the range, and the variant and record constructors are both covariant on the type of the fields, but respectively covariant and contravariant on their label sets.\\
	
	\item\textbf{Multiplicity subtyping} Using an unrestricted ($*$) resource where a linear ($\lin$) one is expected does not compromise safety, provided that, multiplicities aside, the type of the former may safely substitute the type of the latter. We can express this relationship between multiplicities through a preorder captured by inequality $*\sqsubseteq{\lin}$. In our system, function types may be either linear or unrestricted. Thus, type $\tArrow{T_1}{m}{T_2}$ can be considered a subtype of $\tArrow{U_1}{n}{U_2}$ if $U_1$ and $T_2$ are subtypes, respectively, of $T_1$ and $U_2$ and if $m\sqsubseteq n$ (thus we can characterize the function type constructor as covariant on its multiplicity).
\end{itemize}
The rules for our syntactic subtyping relation, interpreted coinductively, are shown in \cref{fig:syntactic-sub}. Rules \textsc{S-Unit}, \textsc{S-Arrow}, \textsc{S-Rcd}, \textsc{S-Vrt}, \textsc{S-RecL} and \textsc{S-RecR} establish the classical subtyping properties associated with both functional and equi-recursive types, with \textsc{S-Arrow} additionally encoding subtyping between linear and unrestricted functions, relying on a preorder on multiplicities also defined in \cref{fig:syntactic-sub}. Rules \textsc{S-End}, \textsc{S-In}, \textsc{S-Out}, \textsc{S-ExtChoice} and \textsc{S-IntChoice} bring to the context-free setting the classical subtyping properties expected from session types, as put forth by Gay and Hole~\cite{DBLP:journals/acta/GayH05}.

The remaining rules account for sequential composition, which distributes over
choice and exhibits a monoidal structure with its neutral
element in $\tSkip$ and left-absorbing element in $\tEnd$. We include, for each session type constructor $S$, a left
rule (denoted by suffix \textsc{L}) of the form $\tSeq{S}{R}\subS S'$ and a
right rule (denoted by suffix \textsc{R}) of the form $S'\subS \tSeq{S}{R}$. An additional rule is necessary for each constructor over which sequential composition does not distribute, associate or neutralize (\textsc{S-InSeq2}, \textsc{S-OutSeq2} and \textsc{S-EndSeq2}). Since we are using a coinductive proof scheme, we include rules to `move' sequential composition down the syntax. Thus, given a type $\tSeq{S}{R}$, we inspect $S$ to decide which rule to apply next.
\begin{restatable}{theorem}{syntacticPreorder}
	\label{thm:syntactic-preorder}
	The syntactic subtyping relation $\subS$ is a preorder on types.
\end{restatable} 
\begin{example}
	Let us briefly return to \cref{ex:intro}. It is now easy to see that $\mathsf{STree} \subS \mathsf{SFullTree1}$: we unfold the left-hand side and apply rule \textsc{S-IntChoice}. Then we apply the distributivity rules as necessary until reaching an internal choice with no continuation, at which point we can apply \textsc{S-IntChoice} again, or until reaching a type with $\tOut{\tInt}$ at the head, at which point we apply \textsc{S-InSeq2}. We repeat this process until reaching $\mathsf{STree}\subS\mathsf{SFullTree0}$, and proceed similarly until reaching $\mathsf{STree}\subS\mathsf{SEmpty}$, which follows from \textsc{S-IntChoice} and \textsc{S-Skip}.
\end{example}
Despite clearly conveying the intended meaning of the subtyping relation, the rules suggest no obvious algorithmic intepretation: on the one
hand, the presence of bare metavariables makes the system not syntax-directed;
on the other hand, rules \textsc{S-RecL}, \textsc{S-RecSeqL} and their
right counterparts lead to infinite derivations which are not solvable by
a conventional fixed-point construction~\cite{DBLP:journals/acta/GayH05,DBLP:books/daglib/0005958}. In
the next section we develop an alternative, semantic approach to subtyping,
which we use as a stepping stone to develop our subtyping algorithm.

\section{Semantic subtyping}\label{sec:semantic-subtyping}
\emph{Semantic equivalence} for context-free session types is usually based on \emph{observational equivalence} or \emph{bisimilarity}, meaning that two session types are considered equivalent if they exhibit exactly the same communication behaviour~\cite{DBLP:conf/icfp/ThiemannV16}. An analogous notion of \emph{semantic subtyping} should therefore rely on an \emph{observational preorder}. In this section we develop such a preorder.

We define the behaviour of types via a labelled transition system (LTS) by establishing relation $\transition{T}{a}{U}$ (``type $T$ transitions by action $a$ to type $U$''). We follow Costa et al.~\cite{DBLP:journals/corr/abs-2203-12877} in attributing behaviour to functional types, allowing them to be encompassed in our observational preorder. The rules defining the trans~ition relation, as well as the grammar that generates all possible transition actions, are shown in \cref{fig:hocfst-lts}. 

In general, each functional type constructor generates a transition for each of
its fields ($\tUnit$ and $\tEnd$, which have none, transition to $\tSkip$). Linear functions, exhibit an additional transition to represent their restricted use (\textsc{L-LinArrow}), and records/variants include a default transition that is independent of their fields (\textsc{L-RcdVrt}). The behaviour of session types is more complex, since it must account for their algebraic properties. Message types exhibit a transition for their payload (\textsc{L-Msg1}, \textsc{L-MsgSeq1}) and another for their continuation, which is $\tSkip$ by omission (\textsc{L-Msg2}, \textsc{L-MsgSeq2}). Choices behave much like records/variants when alone, but are subject to distributivity when composed  (\textsc{L-ChoiceFieldSeq}). Type $\tEnd$, which absorbs its continuation, transitions to $\tSkip$ (\textsc{L-End}, \textsc{L-EndSeq}). Rules \textsc{L-SeqSeq}, \textsc{L-SkipSeq} account for associativity and identity, and rules \textsc{L-Rec} and \textsc{L-RecSeq} dictate that recursive types behave just like their unfoldings. Notice that $\tSkip$ has no transitions.
\begin{figure}[t]
	\text{Labelled transition system}\hfill\fbox{$\transition{T}{a}{T}$}
	\begin{mathpar}
		\inferrule[L-Unit]{}{\transition{\tUnit}{\tUnit}{\tSkip}}
		
		\inferrule[L-ArrowDom]{}{\transition{(\tArrow{T}{m}{U})}{\lArrowDomain}{T}}
		
		\inferrule[L-ArrowRng]{}{\transition{(\tArrow{T}{m}{U})}{\lArrowRange}{U}}
		
		\inferrule[L-LinArrow]{}{\transition{(\tArrow{T}{{\lin}}{U})}{\lArrowLinear}{\tSkip}}
		
		\inferrule[L-RcdVrtField]{k\in L}{\transition{\llparenthesis\ell{:}\,T_\ell\rrparenthesis_{\ell\in L}}{\llparenthesis\rrparenthesis_k}{T_k}}
		
		\inferrule[L-RcdVrt]{}{\transition{\tRcdVrt{\ell}{T_\ell}{L}}{\lRcdVrtDefault}{\tSkip}}
		
		\inferrule[L-Rec]
		{\transition{\subst{x}{\tRec{x}{T}}{T}}{a}{U}}
		{\transition{\tRec{x}{T}}{a}{U}}
		
		\inferrule[L-Msg1]{}{\transition{\tMsg{T}}{\lMsgPayload}{T}}
		
		\inferrule[L-Msg2]{}{\transition{\tMsg{T}}{\lMsgContinuation}{\tSkip}}
		
		\inferrule[L-Choice]{}{\transition{\tChoice{\ell}{S_\ell}{L}}{\lChoiceDefault}{\tSkip}}
		
		\inferrule[L-ChoiceField]{k\in L}{\transition{\tChoice{\ell}{S_\ell}{L}}{\lChoiceLabel{k}}{S_k}}
		
		\inferrule[L-End]{}{\transition{\tEnd}{\tEnd}{\tSkip}}
		
		\inferrule[L-MsgSeq1]{}{\transition{\tSeq{\tMsg{T}}{S}}{\lMsgPayload}{T}}
		
		\inferrule[L-MsgSeq2]{}{\transition{\tSeq{\tMsg{T}}{S}}{\lMsgContinuation}{S}}
		
		\inferrule[L-ChoiceSeq]{}{\transition{\tSeq{\tChoice{\ell}{S_\ell}{L}}{R}}{\lChoiceDefault}{\tSkip}}
		
		\inferrule[L-SkipSeq]{\transition{S}{a}{T}}{\transition{\tSkip;S}{a}{T}}
		
		\inferrule[L-EndSeq]{}{\transition{\tSeq{\tEnd}{S}}{\tEnd}{\tSkip}}
		
		\inferrule[L-SeqSeq]{\transition{S_1;(S_2;S_3)}{a}{T}}{\transition{\tSeq{(\tSeq{S_1}{S_2})}{S_3}}{a}{T}}
		
		\inferrule[L-ChoiceFieldSeq]{k\in L}{\transition{\tSeq{\tChoice{\ell}{S_\ell}{L}}{R}}{\lChoiceLabel{k}}{\tSeq{S_k}{R}}}
		
		\inferrule[L-RecSeq]
		{\transition{\tSeq{(\subst{s}{\tRec{s}{S}}{S})}{R}}{a}{T}}
		{\transition{\tSeq{(\tRec{s}{S})}{R}}{a}{T}}
		
		\inferrule{}{\text{(no rule for $\tSkip$)}}
		\\
		\text{Actions}
		\hfill\\
		a \Coloneqq    \tUnit 
		\svert \lArrowDomain
		\svert \lArrowRange
		\svert \lArrowLinear
		\svert \tEnd
		\svert \lRcdVrtLabel{\ell}
		\svert \lRcdVrtDefault
		\svert \lMsgPayload
		\svert \lMsgContinuation
		\svert \lChoiceDefault
		\svert \lChoiceLabel{\ell}
	\end{mathpar}
	\caption{Labelled transition system. Letters $\mathsf{d}$, $\mathsf{r}$, $\mathsf{p}$, $\mathsf{c}$ in labels stand for ``domain'', ``range'', ``payload'' and ``continuation''.}
	\label{fig:hocfst-lts}
\end{figure}
With the behaviour of types established, we now look for an appropriate notion of observational preorder. Several such notions have been studied in the literature. \emph{Similarity}, defined as follows, is arguably the simplest of them~\cite{DBLP:conf/ijcai/Milner71, DBLP:conf/tcs/Park81}.
\begin{definition}
	\label{simulation}
	A type relation $\mathcal{R}$ is said to be a simulation if, whenever $T\mathcal{R}U$, for all $a$ and $T'$ with $\transition{T}{a}{T'}$ there is $U'$ such that $\transition{U}{a}{U'}$ and $T'\mathcal{R}U'$
	
	Similarity, written $\preceq$, is the union of all simulation relations. We say that a type $U$ simulates type $T$ if $T\preceq U$.
\end{definition} 
Unfortunately, plain similarity is of no use to us. A small example shows why: type $\mathsf{{\oplus}\{\field{A}{\tEnd},\field{B}{\tEnd}\}}$ both simulates and is a subtype of $\mathsf{\oplus\{A{:\,}\tEnd\}}$, while type $\mathsf{\&\{A{:\,}\tEnd\}}$ does not simulate yet is a subtype of $\mathsf{\&\{A{:\,}\tEnd, B{:\,}\tEnd\}}$. Reversing the direction of the simulation would be of no avail either, as it would leave us with the reverse problem.

It is apparent that a more refined notion of simulation is necessary, where the direction of the implication depends on the transition labels. Aarts and Vaandrager provide just such a notion in the form of \emph{$\mathcal{XY}$-simulation}~\cite{DBLP:conf/concur/AartsV10}, a simulation relation parameterized by two subsets of actions, $\mathcal{X}$ and $\mathcal{Y}$, such that actions in $\mathcal{X}$ are simulated from left to right and those in $\mathcal{Y}$ are simulated from right to left, selectively combining the requirements of simulation and reverse simulation.
\begin{definition}
	Let  $\mathcal{X},\mathcal{Y} \subseteq \mathcal{A}$. A type relation $\mathcal{R}$ is said to be an $\mathcal{XY}$-simulation if, whenever $T\mathcal{R}U$, we have: 
	\begin{enumerate}
		\item for each $a \in \mathcal{X}$ and each $T'$ with $\transition{T}{a}{T'}$, there is $U'$ such that $\transition{U}{a}{U'}$ with $T'\mathcal{R}U'$;
		\item for each $a \in \mathcal{Y}$ and each $U'$ with $\transition{U}{a}{U'}$, there is $T'$ such that $\transition{T}{a}{T'}$ with $T'\mathcal{R}U'$.
	\end{enumerate}
	$\mathcal{XY}$-similarity, written $\preceq^{\mathcal{XY}}$, is the union of all $\mathcal{XY}$-simulation relations. We say that a type $T$ is $\mathcal{XY}$-similar to type $U$ if $T\preceq^{\mathcal{XY}}U$.
\end{definition}
Similar or equivalent notions have appeared throughout the literature: \emph{modal refinement}~\cite{DBLP:conf/lics/LarsenT88}, \emph{alternating simulation}~\cite{DBLP:conf/concur/AlurHKV98} and, perhaps more appropriately named (for our purposes), \emph{covariant-contravariant simulation}~\cite{DBLP:conf/calco/FabregasFP09}. Padovani's original subtyping relation for first-order context-free session types~\cite{DBLP:journals/toplas/Padovani19} can also be understood as a refined form of $\mathcal{XY}$-simulation. 

We can tentatively define a semantic subtyping relation ${\lesssim'}$ as $\mathcal{XY}$-similarity, where $\mathcal{X}$ and $\mathcal{Y}$ are the label sets generated by the following grammars for $a_\mathcal{X}$ and $a_\mathcal{Y}$, respectively.
\begin{align*}
	\begin{aligned}[t]
		a_\mathcal{X} &\Coloneqq a_{\mathcal{XY}} \svert \lVrtLabel{\ell} \svert \lExtChoiceLabel{\ell} \\ 
		a_\mathcal{Y} &\Coloneqq a_\mathcal{XY} \svert \lArrowLinear \svert \lRcdLabel{\ell} \svert \lIntChoiceLabel{\ell}
	\end{aligned}
	\qquad 
	\begin{aligned}[t]
		a_\mathcal{XY} \Coloneqq \tUnit \svert \lArrowDomain \svert \lArrowRange \svert \lRcdVrtDefault  \svert \lMsgPayload \svert \lMsgContinuation \svert   \lChoiceDefault \svert  \tEnd
	\end{aligned}
\end{align*}
This would indeed give us the desired result for our previous example, but we
still cannot account for the contravariance of output and function types: we
want $T=\tOut{\mathsf{\{A{:}\,\tInt\}}}$ to be a subtype of
$U=\tOut{\mathsf{\{A{:}\, \tInt, B{:}\, \tBool\}}}$, yet
$T\mathrel{{\lesssim}'}U$ does not hold (in fact, we have
$U\mathrel{{\lesssim}'}T$, a clear violation of run-time safety). The same could
be said for types $\tArrow{\mathsf{\{A{:}\, \tInt\}}}{*}{\tInt}$ and
$\tArrow{\mathsf{\{A{:}\, \tInt, B{:}\, \tBool\}}}{*}{\tInt}$. In short, our
simulation needs the $\lOutPayload$ and $\lArrowDomain$-derivatives to be
related in the direction opposite to that of the initial types. Thus we need to
selectively apply a strong form of \emph{contrasimulation} as well~\cite{C:book/Sangiorgi11,DBLP:conf/concur/Glabbeek93} (the original notion is
defined with weak transitions, a sort of transition we do not address). 

To allow this, we generalize the definition of $\mathcal{XY}$-simulation by parameterizing it on two further subsets of actions and including two more clauses where the direction of the relation between the derivatives is reversed. By analogy with $\mathcal{XY}$-simulation, we call the resulting notion $\mathcal{XYZW}$-simulation.
\begin{definition}
	\label{def:xyzw-sim}
	Let $\mathcal{X,Y,Z,W} \subseteq \mathcal{A}$. A type relation $\mathcal{R}$ is a $\mathcal{XYZW}$-simulation if, whenever $T\mathcal{R}U$, we have:
	\begin{enumerate}
		\item for each $a \in \mathcal{X}$ and each $T'$ with $\transition{T}{a}{T'}$, there is $U'$ such that $\transition{U}{a}{U'}$ with $T'\mathcal{R}U'$;
		\item for each $a \in \mathcal{Y}$ and each $U'$ with $\transition{U}{a}{U'}$, there is $T'$ such that $\transition{T}{a}{T'}$ with $T'\mathcal{R}U'$;
		\item for each $a \in \mathcal{Z}$ and each $T'$ with $\transition{T}{a}{T'}$, there is $U'$ such that $\transition{U}{a}{U'}$ with $U'\mathcal{R}T'$;
		\item for each $a \in \mathcal{W}$ and each $U'$ with $\transition{U}{a}{U'}$, there is $T'$ such that $\transition{T}{a}{T'}$ with $U'\mathcal{R}T'$.
	\end{enumerate}
	$\mathcal{XYZW}$-similarity, written ${\xyzwsim}$, is the union of all $\mathcal{XYZW}$-simulation relations. We say that a type $T$ is $\mathcal{XYZW}$-similar to type $U$ if $T\xyzwsim U$.	
\end{definition}
$\mathcal{XYZW}$-simulation generalizes several existing observational relations: $\mathcal{XY}$-simulation can be defined as an
$\mathcal{XY}\emptyset\emptyset$-simulation, bisimulation as $\mathcal{AA}\emptyset\emptyset$-simulation (alternatively, $\emptyset\emptyset\mathcal{AA}$-simulation or $\mathcal{AAAA}$-simulation), and plain simulation as $\mathcal{A}\emptyset\emptyset\emptyset$-simulation.
\begin{restatable}{theorem}{xyzwsimPreorder}
	\label{thm:xyzw-preorder}
	For any $\mathcal{X,Y,Z,W}$, ${\xyzwsim}$ is a preorder relation on types.
\end{restatable}
Equipped with the notion of $\mathcal{XYZW}$-similarity, we are ready to define the semantic subtyping relation for functional and higher-order context-free session types as follows. 
\begin{definition}
	\label{def:hocfst-sim}
	The semantic subtyping relation for functional and higher-order context-free session types $\simS$ is defined by $T \simS U$ when $T\xyzwsim U$ such that $\mathcal{X}$, $\mathcal{Y}$, $\mathcal{Z}$ and $\mathcal{W}$ are defined as the label sets generated by the following grammars for $a_\mathcal{X}$, $a_\mathcal{Y}$, $a_\mathcal{Z}$ and $a_\mathcal{W}$, respectively.
	\begin{align*}
		\begin{aligned}[t]
			a_\simSX &\Coloneqq a_\mathcal{XY} \svert \lArrowLinear \svert\, \lVrtLabel{\ell}   \svert \lExtChoiceLabel{\ell}\\ 
			a_\simSY &\Coloneqq a_\mathcal{XY} \svert \lRcdLabel{\ell} \svert \lIntChoiceLabel{\ell}
		\end{aligned}
		\qquad 
		\begin{aligned}[t]
			a_\simSZ, a_\simSW &\Coloneqq \lOutPayload \svert \lArrowDomain \\
			a_{\simSX\simSY} &\Coloneqq \tUnit \svert  \lArrowRange \svert \lRcdVrtDefault \svert \lInPayload \svert  \lMsgContinuation \svert \lChoiceDefault \svert  \tEnd
		\end{aligned}
	\end{align*}
\end{definition}
Notice the correspondence between the placement of the labels and the variance of their respective type constructors. Labels arising from covariant positions of the arrow and input type constructors are placed in both the $\mathcal{X}$ and $\mathcal{Y}$ sets, while those arising from the contravariant positions of the arrow and output type constructors are placed in both the $\mathcal{Z}$ and $\mathcal{W}$ sets. Labels arising from the fields of constructors exhibiting width subtyping are placed in a single set, depending on the variance of the constructor on the label set: $\mathcal{X}$ for covariance (external choice and variant constructors), $\mathcal{Y}$ for contravariance (internal choice and record constructors). The function type constructor is covariant on its multiplicity, thus the linear arrow label is placed in $\mathcal{X}$. Finally, default record/variant/choice labels and those arising from nullary constructors are placed in $\mathcal{X}$ and $\mathcal{Y}$, but they could alternatively be placed in $\mathcal{Z}$ and $\mathcal{W}$ or in all four sets (notice the parallel with bisimulation, that can be defined as $\mathcal{AA}\emptyset\emptyset$-simulation, $\emptyset\emptyset\mathcal{AA}$-simulation, or $\mathcal{AAAA}$-simulation).
\begin{example}
	Let us go back once again to our tree serialization example from \cref{sec:introduction}. Here it is also easy to see that $\mathsf{STree}\simS\mathsf{SFullTree1}$. Observe that, on the side of $\mathsf{STree}$, transitions by $\lIntChoiceLabel{\mathsf{Nil}}$ and $\lIntChoiceLabel{\mathsf{Node}}$ always appear together, while on the side of $\mathsf{SFullTree1}$ types transition first by $\lIntChoiceLabel{\mathsf{Node}}$ and then by $\lIntChoiceLabel{\mathsf{Nil}}$. Since $\lIntChoiceLabel{\mathsf{Nil}}$ and $\lIntChoiceLabel{\mathsf{Node}}$ belong exclusively to $\mathcal{Y}$, $\mathsf{STree}$ is always able to match $\mathsf{SFullTree1}$ on these labels (as in all the others in $\mathcal{Y}\cup\mathcal{W}$, and \emph{vice-versa} for $\mathcal{X}\cup\mathcal{Z}$).
\end{example}
\begin{restatable}[Soundness and completeness for subtyping relations]{theorem}{syntacticSemanticSoundnessCompleteness}
	\label{thm:soundness-completeness-stx-sem}
	Let $\vdash T$ and $\vdash U$. Then $T\subS U$ iff $T \simS U$.
\end{restatable}
\section{A subtyping algorithm}\label{sec:algorithm}
The notion of subtyping we have outlined is undecidable. This follows from the fact that our system, albeit different, contains all the features necessary to reconstruct Padovani's proof of undecidability~\cite{DBLP:journals/toplas/Padovani19}. Using just external choices, sequential composition, the $\tSkip$ type and recursion, one is able to encode simple grammars~\cite{DBLP:conf/focs/KorenjakH66} as context-free session types, in a way that language strings correspond to complete LTS traces of types. By exploiting the covariant width-subtyping in external choices, one can show that subtyping for these types corresponds to language inclusion, which is known to be undecidable for simple languages~\cite{DBLP:journals/tcs/Friedman76}.

Despite the undecidability of our subtyping problem, we are still able to devise
a sound (but necessarily incomplete) algorithm for it. In this section we
present this algorithm, an adaptation of the equivalence algorithm of Almeida et
al.~\cite{DBLP:conf/tacas/AlmeidaMV20}. At its core, it determines the
$\mathcal{XYZW}$-similarity of simple grammars. Its application to context-free
session types is facilitated by a translation function to properly encode types
as grammars. The algorithm may likewise be adapted to other domains. Much like the original, our algorithm can be succinctly described in three distinct phases:
\begin{enumerate}
	\item translate the given types to a simple grammar~\cite{DBLP:conf/focs/KorenjakH66} and two starting words;
	\item prune unreachable symbols from productions;
	\item explore an expansion tree rooted at a node containing the initial words, alternating between expansion and simplification operations until either an empty node is found (decide \True) or all nodes fail to expand (decide \textbf{False}).
\end{enumerate}
\subparagraph*{Phase 1}
The first phase consists of translating the two types to a grammar in \emph{Greibach normal form}  (GNF)~\cite{DBLP:journals/jacm/Greibach65}, i.e., a grammar where all productions have the form $Y\to a \vec{Z}$, and two starting words $(\vec X, \vec Y)$. A word is defined as a sequence of non-terminal symbols. We can check the $\mathcal{XYZW}$-similarity of words in GNF grammars because they naturally induce a labelled transition system, where states are words $\vec{X}$, actions are terminal symbols $a$ and the transition relation is defined as $\transitionGNF{\prods}{X\vec{Y}}{a}{\vec{Z}\vec{Y}}$ when $X \to a \vec{Z} \in \prods$. We denote the bisimilarity and $\mathcal{XYZW}$-similarity of grammars by, respectively, ${\sim_\prods}$ and ${\xyzwsimP}$, where $\mathcal{P}$ is the set of productions. We also let ${\simS_\prods}$ denote grammar $\mathcal{XYZW}$-similarity with label sets as in \cref{def:hocfst-sim}. The deterministic nature of context-free session types allows their corresponding grammars to be simple~\cite{DBLP:conf/focs/KorenjakH66}: for each non-terminal $Y$ and terminal symbol $a$, we have at most one production of the form $Y \to a \vec{Z}$.

The grammar translation procedure $\grm$ remains unchanged from the original equivalence algorithm~\cite{DBLP:conf/tacas/AlmeidaMV20}, and for this reason we omit its details (which include generating productions for all $\mu$-subterms in types). However, this procedure relies on two auxiliary definitions which must be adapted: the $\unravel$ function (\cref{def:unr}), which normalizes the head of session types and unravels recursive types until reaching a type constructor, and the $\word$ procedure (\cref{def:word}), which builds a word from a session type while updating a set $\prods$ of productions. 
\begin{definition}
	\label{def:unr}
	The \emph{unraveling} of a type $T$ is defined by induction on the structure of $T$:
	\begin{align*}
		\begin{aligned}[c]
			\unravel(\tRec{x}{T}) &= \unravel(\subst{x}{\tRec{x}{T}}{T})\\
			\unravel(\tSeq{\tEnd}{S}) &= \tEnd\\
			\unravel(\tSeq{\tChoice{\ell}{S_\ell}{L}}{R}) &= \tChoice{\ell}{S_\ell;R}{L}\\
		\end{aligned}
		\quad 
		\begin{aligned}[c]
			\unravel(\tSeq{\tSkip}{S}) &= \unravel(S)\\
			\unravel(\tSeq{(\tRec{s}{S})}{R}) &= \unravel(\tSeq{(\subst{s}{\tRec{s}{S}}{S})}{R})\\
			\unravel(\tSeq{(\tSeq{S_1}{S_2})}{S_3}) &= \unravel(\tSeq{S_1}{(\tSeq{S_2}{S_3})})
		\end{aligned}
	\end{align*}
	and in all other cases by $\unravel(T)=T$.
\end{definition}
\begin{definition}
	\label{def:word}
	The word corresponding to a well-formed type $T$, $\word(T)$, is built by descending on the structure of $T$ while updating a set $\prods$ of productions:
	\begin{align*}
		\word(\tUnit) =& \,Y\text{, setting}\,\, \prods \coloneqq \prods\cup\{Y \to \tUnit\}\\
		\word(\tArrow{U}{{\lin}}{V}) =& \,Y \text{, setting}\,\,\prods \coloneqq \prods\cup\{Y\to {\lArrowDomain} \word(U),Y\to {\lArrowRange} \word(V), Y\to {\lArrowLinear}\} \\
		\word(\tArrow{U}{\mathbf{*}}{V}) =& \,Y \text{, setting}\,\,\prods \coloneqq \prods\cup\{Y\to {\lArrowDomain} \word(U),Y\to {\lArrowRange} \word(V) \} \\
		\word(\tRcdVrt{\ell}{T_\ell}{L}) =& \,Y \text{, setting}\,\,\prods \coloneqq \prods\cup\{Y \to \lRcdVrtDefault \bot\}\cup \{Y \to \llparenthesis\rrparenthesis_k \word(T_k) \svert k \in L\}\\
		\word(\mathtt{\tSkip}) =& \, \varepsilon\\
		\word(\mathtt{\tEnd}) =& \, Y \text{, setting}\,\,\prods \coloneqq \prods\cup\{Y \to \tEnd \bot\}\\
		\word(\inout U) =& \, Y \, \text{, setting}\,\,\prods \coloneqq \prods\cup\{Y \to \lMsgPayload \word(U) \bot,Y \to \lMsgContinuation\}\\
		\word(\tChoice{\ell}{S_\ell}{L}) =&\, Y \, \text{, setting}\,\,\prods \coloneqq \prods\cup\{Y \to \lChoiceDefault \bot\}\cup \{Y \to \lChoiceLabel{k} \word(S_k) \svert k \in L\}\\
		\word(S_1;S_2) =& \, \word(S_1)\word(S_2)\\
		\word(\tRec{x}{U}) =& \, X
	\end{align*}
	where, in each equation, $Y$ is understood as a fresh non-terminal symbol,
	$X$ as the non-terminal symbol corresponding to type reference $x$, and
	$\bot$ as a non-terminal symbol without productions.
\end{definition}
\begin{example}\label{ex:grammar}
	Consider again the types for tree serialization in \cref{sec:introduction}. Suppose we want to know whether $\tArrow{\mathsf{SFullTree0}}{\un}{\tUnit} \simS \tArrow{\mathsf{STree}}{\lin}{\tUnit}$. We know that the grammar generated for these types is as follows, with $X_0$ and $Y_0$ as their starting words. 
	\begin{align*}
		\begin{aligned}[t]
			X_0 &\to \lArrowDomain X_1 \\
			X_0 &\to \lArrowRange X_5 \\ 
			X_1 &\to \lIntChoiceLabel{\mathsf{Node}} X_2 X_3 X_2\\
			X_1 &\to \lIntChoiceDefault \bot
		\end{aligned}
		\quad
		\begin{aligned}[t]
			X_2 &\to \lIntChoiceLabel{\mathsf{Empty}}\\
			X_2 &\to \lIntChoiceDefault \bot\\
			X_3 &\to \lOutPayload X_4 \bot \\
			X_3 &\to \lOutContinuation 
		\end{aligned}
		\quad 
		\begin{aligned}[t]
			X_4 &\to \tInt \\
			X_5 &\to \tUnit \\
		\end{aligned}
		\quad 
		\begin{aligned}[t]
			Y_0 &\to \lArrowDomain Y_1\\
			Y_0 &\to \lArrowRange X_5\\ 
			Y_0 &\to \lArrowLinear\\
		\end{aligned}
		\quad 
		\begin{aligned}[t]
			Y_1 &\to \lIntChoiceDefault \bot \\
			Y_1 &\to \lIntChoiceLabel{\mathsf{Empty}}\\
			Y_1 &\to \lIntChoiceLabel{\mathsf{Node}}Y_1 X_3 Y_1
		\end{aligned}
	\end{align*}
\end{example}
For the rest of this section let $\wellFormed{}{T}$, $\wellFormed{}{U}$,
$(\vec{X}_T,\prods') = \grm(T,\emptyset)$ and
$(\vec{X}_U,\prods)=\grm(U,\prods')$.
\begin{restatable}[Soundness for grammars]{theorem}{sfg}
	\label{thm:soundness-grammars}
	If $\vec{X}_T \simS_\prods \vec{X}_U$, then $T \simS U$.
\end{restatable}
\subparagraph*{Phase 2} The grammars generated by procedure $\grm$ may contain
unreachable words, which can be ignored by the algorithm. Intuitively, these
words correspond to communication actions that cannot be fulfilled, such as the part $\tIn{\tBool}$ in type $\tSeq{(\tRec{s}{\tSeq{\tOut{\tInt}}{s}})}{\tIn{\tBool}}$. Formally, these words appear in productions following what are known as \textit{unnormed words}.
\begin{definition} 
	Let $\vec a$ be a non-empty sequence of non-terminal symbols $a_1,\ldots,a_n$. Write $\transitionGNF{\prods}{\vec Y}{\vec a}{\vec Z}$ when $\transitionGNF{\prods}{\vec Y}{a_1}{\transitionGNF{\prods}{\ldots}{a_n}{\vec Z}}$. We say that a word $\vec Y$ is \textit{normed} if $\transitionGNF{\prods}{\vec Y}{\vec a}{\varepsilon}$ for some $\vec a$, and \textit{unnormed} otherwise. If $\vec Y$ is normed and $\vec a$ is the shortest path such that $\transitionGNF{\prods}{\vec Y}{\vec a}{\varepsilon}$, then $\vec a$ is called the \textit{minimal path} of $\vec Y$, and its length is the \textit{norm} of $\vec Y$, denoted $|\vec Y|$.
\end{definition}
It is known that any unnormed word $\vec Y$ is bisimilar to its concatenation with any other word, i.e., if $\vec Y$ is unnormed, then $\vec Y \sim_\prods \vec Y \vec X$. It is also easy to show that ${\sim_\prods} \subseteq {\simS_\prods}$, and hence that $\vec Y \simS_\prods \vec Y \vec X$. In this case, $\vec X$ is said to be unreachable and can be safely removed from the grammar. We call the procedure of removing all unreachable symbols from a grammar \textit{pruning}, and denote the pruned version of a grammar $\prods$ by $\prune(\prods)$.

\begin{restatable}[Pruning preserves $\mathcal{XYZW}$-similarity]{lemma}{pruningPreservesXYZWSimilarity}
	\label{thm:pruning-preserves-xyzw}
	$\vec X \xyzwsim_\prods \vec Y$ iff $\vec X \xyzwsim_{\prune(\prods)} \vec Y$
\end{restatable}

\subparagraph*{Phase 3} In its third and final phase, the algorithm explores an
\emph{expansion tree}, alternating between expansion and simplification steps.
An expansion tree is a tree whose nodes are sets of pairs of words, whose root
is the singleton set containing the pair of starting words under test, and where
every child is an \emph{expansion} of its parent. A branch is deemed
\textit{successful} if it is infinite or has an empty leaf, and deemed
\textit{unsuccessful} otherwise. The original definition of expansion ensures
that the union of all nodes along a successful branch (without simplifications) constitutes a bisimulation~\cite{DBLP:conf/concur/JancarM99}. We adapt this definition to ensure that such
a union yields an $\mathcal{XYZW}$-simulation instead.
\begin{definition}
	\label{def:xyzw-expansion}
	The \emph{$\mathcal{XYZW}$-expansion} of a node $N$ is defined as the minimal set $N'$ such that, for every pair $(\vec{X},\vec{Y})$ in $N$, it holds that:
	\begin{enumerate}
		\item if   $\vec{X}\to a\vec{X'}$ and $a \in \mathcal{X}$
		then $\vec{Y}\to a \vec{Y'}$ 
		with $(\vec{X'},\vec{Y'}) \in N'$
		\item if   $\vec{Y}\to a\vec{Y'}$ and $a \in \mathcal{Y}$
		then $\vec{X}\to a \vec{X'}$ 
		with $(\vec{X'},\vec{Y'}) \in N'$
		\item if   $\vec{X}\to a\vec{X'}$ and $a \in \mathcal{Z}$
		then $\vec{Y}\to a \vec{Y'}$ 
		with $(\vec{Y'},\vec{X'}) \in N'$
		\item if   $\vec{Y}\to a\vec{Y'}$ and $a \in \mathcal{W}$
		then $\vec{X}\to a \vec{X'}$ 
		with $(\vec{Y'},\vec{X'}) \in N'$
	\end{enumerate}
\end{definition}

\begin{restatable}[Safeness property for $\mathcal{XYZW}$-simulation]{lemma}{sp}
	\label{lemma:safeness}
	Given a set of productions $\prods$, $\vec X \xyzwsim_\prods \vec Y$ iff the expansion tree rooted at $\{( \vec X, \vec Y )\}$ has a successful branch.
\end{restatable}
The simplification stage consists of applying rules that safely modify the expansion tree during its construction, in an attempt to keep some branches finite. The rules are iteratively applied to each node until a fixed point is reached, at which point we can proceed with expansion. To each node $N$ we apply three simplification rules, adapted from the equivalence algorithm~\cite{DBLP:conf/tacas/AlmeidaMV20}:
\begin{enumerate} 
	\item\textsc{Reflexivity}: omit pairs of the form $(\vec X, \vec X)$;
	\item\label{rule:cong}\textsc{Preorder}: omit pairs belonging to the least preorder containing the ancestors of $N$;
	\item\textsc{Split}: if $(X_0 \vec X, Y_0 \vec Y) \in N$ and $X_0$ and $Y_0$ are normed, then:
	\begin{itemize}
		\item Case $|X_0| \leq |Y_0|$: Let $\vec a$ be a minimal path for $X_0$ and $\vec Z$ the word such that $\transitionGNF{\prods}{Y_0}{\vec a}{\vec Z}$. Add a sibling node for $N$ including pairs $(X_0 \vec Z, Y_0)$ and $(\vec X, \vec Z \vec Y)$ in place of $(X_0 \vec X, Y_0 \vec Y)$;
		\item Otherwise: Let $\vec a$ be a minimal path for $Y_0$ and $\vec Z$ the word such that $\transitionGNF{\prods}{X_0}{\vec a}{\vec Z}$. Add a sibling node for $N$ including pairs $(X_0, Y_0 \vec Z)$ and $(\vec Z \vec X, \vec Y)$ in place of $(X_0 \vec X, Y_0 \vec Y)$.
	\end{itemize}
\end{enumerate}
When a node is simplified, we keep track of the original node in a sibling, thus ensuring that along the tree we keep an ``expansion-only'' branch.

The algorithm explores the tree by breadth-first search using a queue of node-ancestors pairs, thus avoiding getting stuck in infinite branches, and alternates between expansion and simplification steps until it terminates with $\False$ if all nodes fail to expand or with $\True$ if an empty node is reached. The following pseudo-code illustrates the procedure.
\begin{align*}
	&\subG(\vec X, \vec Y , \prods) = \fun{explore}(\singletonQueue((\{(\vec X, \vec Y )\}, \emptyset), \prods) \\
	&\quad\textbf{where }\fun{explore}(q, \prods) = \\
	&\qquad \textbf{if}\, \emptyf(q) \,\textbf{then}\,\,\False \text{  \textcolor{lipicsLineGray}{\% all nodes failed to expand}}\\
	&\qquad \textbf{else let}\,  (n, a) = \front(q)\textbf{ in}\\
	&\qquad\quad           \textbf{if}\, \emptyf(n) \,\textbf{then} \,\,\True \text{  \textcolor{lipicsLineGray}{\% empty node reached}}\\
	&\qquad\quad          \textbf{else if}\,\, \fun{hasExpansion}(n, \prods) \text{ \textcolor{lipicsLineGray}{\% then expand, simplify and recur}}\\
	&\qquad\qquad\quad\: \textbf{then}\, \fun{explore}(\simplify(\fun{expand}(n, \prods), a \cup n, \dequeue(q)), \prods) \\
	&\qquad\qquad\quad\:  \textbf{else}\, \fun{explore}(\dequeue(q), \prods)  \text{ \textcolor{lipicsLineGray}{\% otherwise, discard node}}
\end{align*}
\begin{example}
	The $\mathcal{XYZW}$-expansion tree for \cref{ex:grammar} is illustrated in \cref{fig:expansion-tree}.
\end{example}
\begin{figure}[t]
	\centering
	\includegraphics[width=0.69\textwidth]{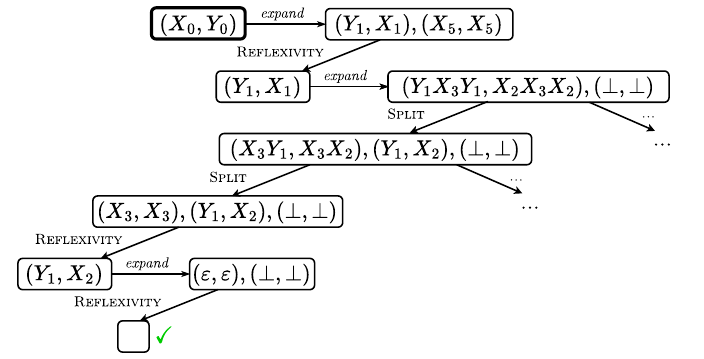}
	\caption{An $\mathcal{XYZW}$-expansion tree for \cref{ex:grammar}, exhibiting a finite successful branch.}
	\label{fig:expansion-tree}
\end{figure}
Finally, function $\subT$ puts all the pieces of the algorithm together: 
\begin{align*}
	&\subT(T,U) = \textbf{let }(\vec X, \prods') = \grm(T,\emptyset), (\vec Y, \prods) = \grm(U, \prods')\textbf{ in }\subG(\vec X, \vec Y, \prune(\prods))
\end{align*}
It receives two well-formed types $T$ and $U$, computes their grammar and respective starting words $\vec X$ and $\vec Y$, prunes the productions of the grammar and, lastly, uses function $\subG$ to determine whether $\vec X \simS_\prods \vec Y$. 

The following result shows that algorithm $\subT$ is sound with respect to semantic subtyping relation on functional and higher-order context-free session types. 

\begin{restatable}[Soundness]{theorem}{soundness}
	\label{thm:soundness-alg}
	If $\subT(T, U)$ returns \True, then $T \simS U$.
\end{restatable}

\section{Evaluation}\label{sec:evaluation}

We have implemented our subtyping algorithm in Haskell and integrated it in the freely available compiler for FreeST, a statically typed functional programming language featuring message-passing channels governed by context-free session types~\cite{DBLP:journals/iandc/AlmeidaMTV22,DBLP:journals/corr/abs-1904-01284,almeida_mordido_vasconcelos_2019}. The FreeST compiler features a running implementation of the type equivalence algorithm of Almeida et al.~\cite{DBLP:conf/tacas/AlmeidaMV20}. With our contributions, FreeST effectively gains support for subtyping at little to no cost in performance. In this section we present an empirical study to support this claim.

We employed three test suites to evaluate the performance of our algorithm: a suite of handwritten pairs of types, a suite of randomly generated pairs of types, and a suite of handwritten FreeST programs. We focus on the last two, since they allow a more robust and realistic analysis. All data was collected on a machine featuring an Intel Core i5-6300U at 2.4GHz with 16GB of RAM.

To build our randomly generated suite we employed a type generation module, implemented using the Quickcheck library~\cite{DBLP:conf/icfp/ClaessenH00} and following an algorithm induced from the properties of subtyping, much like the one induced by Almeida et al.~\cite{DBLP:conf/tacas/AlmeidaMV20} from the properties of bisimilarity. It includes generators for valid and invalid subtyping pairs. We conducted our evaluation by taking the running time of the algorithm on 2000 valid pairs and 2000 invalid pairs, ranging from 2 to 730 total AST nodes, with a timeout of 30s (ensuring it terminates with either \textbf{True}, \textbf{False} or \textbf{Unknown}). The results are plotted in \cref{fig:performance}. Despite the incompleteness of the algorithm, we encountered no false negatives, but obtained 188 timeouts. We found, as expected, that the running time increases considerably with the number of nodes. When a result was produced, valid pairs took generally longer.

\begin{figure}[t]
	\centering
	\subfloat[Performance on valid and invalid subtyping pairs\label{fig:performance}]{%
		\includegraphics[width=0.45\textwidth]{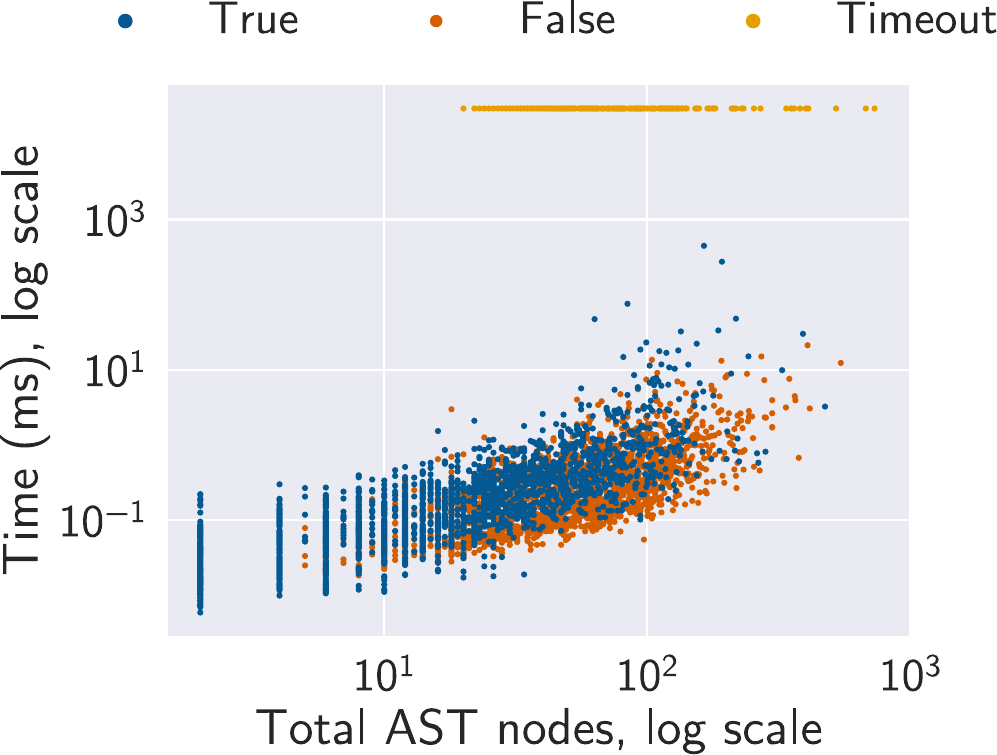}%
	}\hfil
	\subfloat[Performance comparison against the original equivalence algorithm\label{fig:comparison}]{%
		\includegraphics[width=0.425\textwidth]{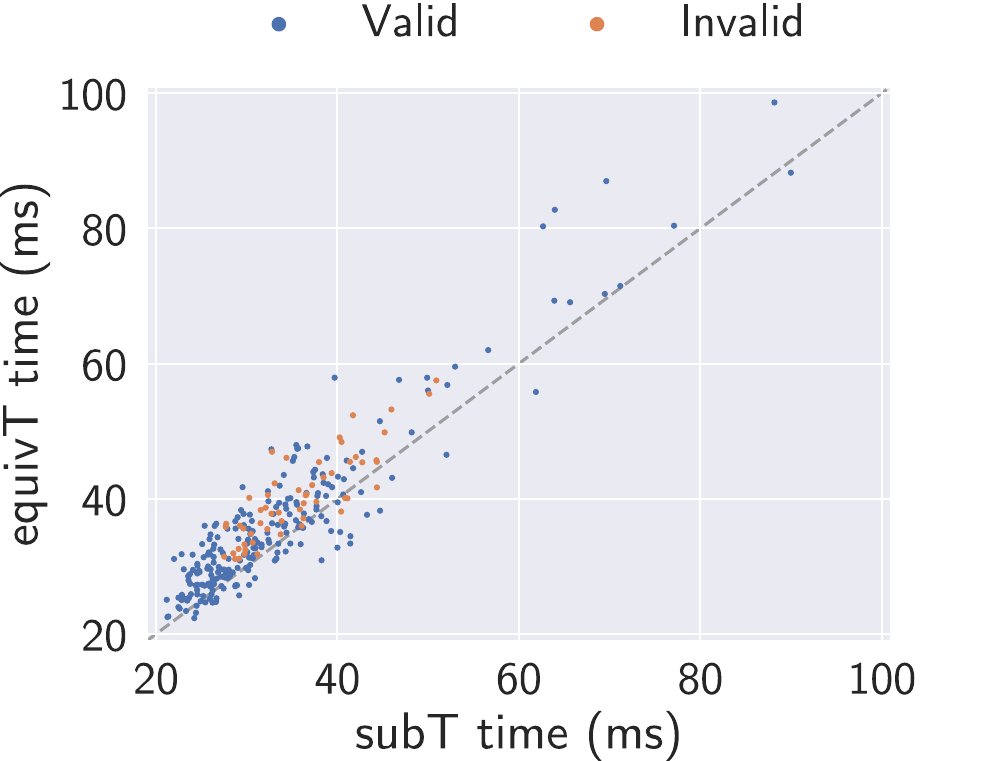}%
	}
	
	\caption{Performance evaluation and comparison.}
	\label{fig:image2}
	
\end{figure}

Randomly generated types allow for a robust analysis, but they typically do not reflect the types encountered by a subtyping algorithm in its most obvious practical application, a~compiler. For this reason, we turn our attention to our suite of FreeST programs, comprised of 286 valid and invalid programs collected throughout the development of the FreeST language. Programs range from small examples demonstrating particular features of the language to concurrent applications simulating, for example, an FTP server. 

We began by integrating the algorithm in the FreeST compiler, placing next to every call to the original algorithm~\cite{DBLP:conf/tacas/AlmeidaMV20} (henceforth $\mathsf{equivT}$) a call to $\mathsf{subT}$ on the same pairs of types. We then ran each program in our suite 10 times, collecting and averaging the accumulated running time of both algorithms on the same pairs of types. We then took the difference between the average accumulated running times of $\mathsf{subT}$ and $\mathsf{equivT}$, obtaining an average difference of -3.85ms, with a standard deviation of 7.08ms, a minimum difference of -71.29ms and a maximum difference of 8.03ms ($\mathsf{subT}$ performed faster, on average). \cref{fig:comparison} illustrates this comparison by plotting against each other the accumulated running times (for clarity, those in the 20-100ms range) of both algorithms  during the typechecking phase of each.

The data collected in this evaluation suggests that replacing the original equivalence algorithm~\cite{DBLP:conf/tacas/AlmeidaMV20} with the subtyping algorithm in the FreeST typechecker generally does not incur an overhead, while providing additional expressive power for programmers.

\section{Related work}\label{sec:related-work}
Session types emerged as a formalism to express communication protocols and
statically verify their implementations~\cite{DBLP:conf/concur/Honda93,DBLP:conf/esop/HondaVK98}.  Initial formulations
allowed only pairwise, tail-recursive protocols, earning such types the `binary' and `regular' epithets. Since then, considerable efforts have been made to extend the theory of session types beyond the binary and regular realms: multiparty session types allow sessions with multiple participants~\cite{DBLP:conf/popl/HondaYC08}, while context-free session types~\cite{DBLP:conf/icfp/ThiemannV16} and nested session types~\cite{DBLP:conf/esop/DasDMP21} allow non-regular communication patterns. Our work is centered on context-free session types, which have seen considerable development since their introduction, most notably their integration in System F~\cite{DBLP:journals/iandc/AlmeidaMTV22, DBLP:conf/esop/PocasCMV23}, an higher-order formulation~\cite{DBLP:journals/corr/abs-2203-12877}, as well as proposals for kind and type inference~\cite{DBLP:journals/corr/abs-2304-06396,DBLP:journals/toplas/Padovani19}.

Subtyping is a standard feature of many type systems, and the literature on the topic is vast~\cite{Amadio93toplas,Brandt98fi,Castagna05ppdp,Castagna14popl,Danielsson10mpc,Dolan17phd,lakhani2022polarized}. Its conventional interpretation, based on the notion of substitutability, originates from the work of Liskov~\cite{DBLP:conf/oopsla/Liskov87}. Multiple approaches to subtyping for regular session types have been proposed, and they can be classified according to the objects they consider substitutable: channels \textit{versus} processes (the difference being most notable in the variance of type constructors).
The earliest approach, subscribing to the substitutability of channels, is that of Gay and Hole~\cite{DBLP:journals/acta/GayH05}. It is also the one we follow. A later formulation, proposed by Carbone et al.~\cite{DBLP:conf/esop/CarboneHY07}, subscribes to the substitutability of processes. A survey of both interpretations is given by Gay~\cite{DBLP:conf/birthday/Gay16}. The interaction between subtyping and polymorphism for regular session types, in the form of bounded quantification, has been investigated by Gay~\cite{DBLP:journals/mscs/Gay08}.  Horne and Padovani study subtyping under the linear logic interpretation of regular session types~\cite{DBLP:journals/corr/abs-2304-06398}, showing that it preserves termination of processes.

Subtyping for session types has spread beyond the regular realm. Das et al.~\cite{DBLP:conf/esop/DasDMP21} introduce subtyping for nested session types, show the problem to be undecidable and present a sound but incomplete algorithm. In the context-free setting, the first and, to the best of our knowledge, only formulation before our work is that of Padovani~\cite{DBLP:journals/toplas/Padovani19}. It proposes a simulation-based subtyping relation, proves the undecidability of the subtyping problem and provides a sound but incomplete algorithm. This undecidability proof also applies to our system, as it possesses all the required elements: width-subtyping on choices, sequential composition and recursion. The subtyping relation proposed by Padovani contemplates neither input/output subtyping nor functional subtyping. Furthermore, its implementation relies on the subtyping features of OCaml, the implementation language. In contrast, we propose a more expressive relation, featuring input/output subtyping, as well as functional subtyping. Furthermore, we provide an also sound algorithm that is independent of the implementation language.

Our subtyping relation is based on a novel form of observational preorder, $\mathcal{XYZW}$-simulation. There is, as far as we know, no analogue in the literature. It is a generalization of $\mathcal{XY}$-simulation, introduced by Aarts and Vaandrager in the context of learning automata~\cite{DBLP:conf/concur/AartsV10} but already known, under slightly different forms, as modal refinement~\cite{DBLP:conf/lics/LarsenT88}, alternating simulation~\cite{DBLP:conf/concur/AlurHKV98} and covariant-contravariant simulation~\cite{DBLP:conf/calco/FabregasFP09}. The contravariance on the derivatives introduced by $\mathcal{XYZW}$-simulation is also prefigured in contrasimulation~\cite{C:book/Sangiorgi11,DBLP:conf/concur/Glabbeek93}, but the former uses strong transitions whereas the latter uses weak ones. There is a vast literature on other observational relations, to which Sangiorgi's book provides an overview~\cite{C:book/Sangiorgi11}. 

Our algorithm decides the $\mathcal{XYZW}$-similarity of simple grammars~\cite{DBLP:conf/focs/KorenjakH66}. It is an adaptation of the bisimilarity algorithm for simple grammars of Almeida et al.~\cite{DBLP:conf/tacas/AlmeidaMV20}. To our knowledge, these are the only running algorithms of their sort. 
Henry and S{\'e}nizergues~\cite{10.1007/978-3-642-39274-0_16} proposed an algorithm to decide the language equivalence problem on deterministic pushdown automata.
On the related topic of basic process algebra (BPA), BPA processes have been shown to be equivalent to grammars in GNF~\cite{DBLP:journals/jacm/BaetenBK93}, of which simple grammars are a particular case. This makes results and algorithms for BPA processes applicable to grammars in GNF, and \textit{vice-versa}. A bisimilarity algorithm for general BPA processes, of doubly-exponential complexity, has been proposed by Burkart et al.~\cite{DBLP:conf/mfcs/BurkartCS95}, while an analogous polynomial-time algorithm for the special case of normed BPA processes has been proposed by Hirschfield et al.~\cite{DBLP:journals/tcs/HirshfeldJM96}.

\section{Conclusion and future work}\label{sec:conclusion}
We have proposed an intuitive notion of subtyping for context-free session types, based on a novel form of observational preorder, $\mathcal{XYZW}$-simulation. This preorder inverts the direction of the simulation in the derivatives covered by its $\mathcal{W}$ and $\mathcal{Z}$ parameters, allowing it to handle co/contravariant features of input/output types. We take advantage of the fact that $\mathcal{XYZW}$-simulation generalizes bisimulation to derive a sound subtyping algorithm from an existing type equivalence algorithm. 

Despite its unavoidable incompleteness, stemming from the undecidability of our notion of subtyping, our algorithm has not yielded any false negatives. Thus, we conjecture that is partially correct: it may not halt, but, when it does, the answer is correct. We cannot, however, back this claim without a careful analysis of completeness and termination, which we leave for future work. We believe such an analysis will advance the understanding of the subtyping problem by clarifying the practical reasons for its undecidability.

As shown by Thiemann and Vasconcelos~\cite{DBLP:conf/icfp/ThiemannV16}, support for polymorphism and polymorphic recursion is paramount in practical applications of context-free session types. Exploring the interaction between polymorphism and subtyping in the context-free setting, possibly in the form of \textit{bounded quantification}, is therefore another avenue for future work.


\bibliography{main}

\newpage 

\appendix

\section{Preliminaries}
\subsection{Type formation}
\label{subsec:type-formation}
Here we introduce the rules for type formation, which ensure types are closed (no free references) and contractive. For convenience, type formation also ensures that types under a $\mu$ binder are not equivalent to $\tSkip$ and that all type references introduced by such binders are pairwise distinct. 

The type formation judgement $\wellFormed{\Delta}{T}$ (``$T$ is well-formed under context $\Delta$'') is defined by the rules in \cref{fig:isterminated} (we understand the notation $\Delta,x$ as requiring $x \notin \Delta$). These rules depend on the contractivity judgement $\contr{T}{x}$ (``$T$ is contractive on reference $x$''), defined by the rules also in \cref{fig:isterminated}. The rules for contractivity depend, in turn, on judgment $\terminated{T}$ (``$T$ is terminated'') to characterize types that exhibit no communication action. The rules for this judgment can also be found in \cref{fig:isterminated}.
\begin{figure}[t]
	\text{Type formation (\emph{inductive})}\hfill\fbox{$\Delta \vdash T$}
	\begin{mathpar}
		\inferrule[TF-Axiom]
		{T=\tUnit,\tEnd,\tSkip}
		{\wellFormed{\Delta}{T}} 
		
		\inferrule[TF-Arrow]
		{\wellFormed{\Delta}{T} \\ \wellFormed{\Delta}{U} }
		{\wellFormed{\Delta}{\tArrow{T}{m}{U}}}
		
		\inferrule[TF-RcdVrt]
		{\wellFormed{\Delta}{T_\ell} \\ (\forall \ell \in L)}
		{\wellFormed{\Delta}{\tRcdVrt{\ell}{T_\ell}{L}}}
		
		\inferrule[TF-Msg]
		{\wellFormed{\Delta}{T}}
		{\wellFormed{\Delta}{\tMsg{T}}}
		
		\inferrule[TF-Choice]
		{\wellFormed{\Delta}{S_\ell} \\ (\forall \ell \in L)}
		{\wellFormed{\Delta}{\tChoice{\ell}{S_\ell}{L}}}
		
		\inferrule[TF-Seq]
		{\wellFormed{\Delta}{S} \\ \wellFormed{\Delta}{R}}
		{\wellFormed{\Delta}{\tSeq{S}{R}}}
		
		\inferrule[TF-Var]
		{x \in \Delta}
		{\wellFormed{\Delta}{x}}
		
		\inferrule[TF-Rec]
		{\notTerminated{T} \\ \contr{T}{x} \\ \wellFormed{\Delta,x}{T}}
		{\wellFormed{\Delta}{\tRec{x}{T}} }
	\end{mathpar}
	\text{Contractivity (\emph{inductive})}\hfill\fbox{$\contr{T}{x}$}
	\begin{mathpar}
		\inferrule[C-Axiom]
		{T=\tUnit,U\overset{m}{\to}V,\llparenthesis\ell:T_\ell\rrparenthesis_{\ell\in L},\tEnd,\inout T,\tChoice{\ell}{S_\ell}{L},\tSkip}
		{\contr{T}{x}}
		\\
		\inferrule[C-Seq1]
		{\terminated{S} \\ \contr{R}{x}}
		{\contr{\tSeq{S}{R}}{x}}
		
		\inferrule[C-Seq2]
		{\notTerminated{S} \\ \contr{S}{x}}
		{\contr{\tSeq{S}{R}}{x}}
		
		\inferrule[C-Var]
		{y \neq x}
		{\contr{y}{x}}
		
		\inferrule[C-Rec]
		{\contr{T}{x}}
		{\contr{\tRec{y}{T}}{x}}
	\end{mathpar}
	\text{Is-terminated (\emph{inductive})}\hfill\fbox{$T\checkmark$}
	\begin{mathpar}
		\inferrule[\checkmark-Skip]
		{}
		{\terminated{\tSkip}}
		
		\inferrule[\checkmark-Seq]
		{\terminated{S} \\ \terminated{R}}
		{\terminated{\tSeq{S}{R}}}
		
		\inferrule[\checkmark-SRec]
		{\terminated{S}}
		{\terminated{\tRec{s}{S}}}
	\end{mathpar}
	\caption{Type formation.}
	\label{fig:isterminated}
\end{figure}
\subsection{Substitution}
The following lemma shows that we can discard type references from type formation contexts, under the assumption that they do not occur free in the type in question.
\begin{lemma}[Type strengthening]
	Let $x\notin\free(T)$. If $\wellFormed{\Delta,x}{T}$, then $\wellFormed{\Delta}{T}$.
\end{lemma}
\begin{proof}
	By rule induction on the hypothesis.
\end{proof}
The following lemma shows that substitution preserves the good properties of types: termination, contractivity and type formation. From it follows that these properties are also preserved by the $\unravel$ function (\cref{def:unr}).
\begin{lemma}[Type substitution]
	\label{lemma:type-substitution}
	Suppose that $\Delta \vdash U$.
	\begin{enumerate}
		\item If $\wellFormed{\Delta,x}{T}$ and $\terminated{T}$, then $\terminated{\subst{x}{U}{T}}$.
		\item If $\wellFormed{\Delta,x}{T}$ and $\notTerminated{T}$, then $\notTerminated{\subst{x}{U}{T}}$.
		\item If $\Delta,x,y \vdash T$ and $\contr{T}{y}$ and $y\notin\mathsf{free}(U)$, then $\contr{\subst{x}{U}{T}}{y}$.
		\item If $\wellFormed{\Delta,x}{T}$ then $\wellFormed{\Delta}{\subst{x}{U}{T}}$.
	\end{enumerate}
\end{lemma}
\begin{proof}
	\hfill 
	\begin{enumerate}
		\item \label{proof:lemma:type-substitution:term} By rule induction on $\terminated{T}$.
		\item \label{proof:lemma:type-substitution:notterm} By structural induction on $T$. All cases are either straightforward or follow from the induction hypothesis.
		\item \label{proof:lemma:type-substitution:contr} By rule induction on $\contr{T}{y}$, using \cref{proof:lemma:type-substitution:term} and \cref{proof:lemma:type-substitution:notterm}. All cases follow from the induction hypothesis except the case for rule \textsc{C-Var}, where we have $T=z$ with $z\neq x,y$, and where the result follows from hypothesis $\contr{z}{y}$.
		\item \label{proof:lemma:type-substitution:tf} By rule induction on $\Delta,x\vdash T$. For the case \textsc{TF-Rec} we have $T=\tRec{y}{V}$. The premises to the rule are $\notTerminated{V}$, $\contr{V}{y}$ and $\Delta,x,y\vdash V$. Induction on the third premise gives $\Delta,y\vdash \subst{x}{U}{V}$. \cref{proof:lemma:type-substitution:notterm} gives $\notTerminated{\subst{x}{U}{V}}$, while \cref{proof:lemma:type-substitution:contr} gives $\contr{\subst{x}{U}{V}}{y}$. Rule \textsc{TF-Rec} gives $\Delta \vdash \tRec{y}{(\subst{x}{U}{V})}$. Conclude with the definition of substitution. For the case \textsc{TF-Var} with $T=y\neq x$, we have $x \notin \mathsf{free}(y)$. The result follows from hypothesis $\wellFormed{\Delta,x}{y}$ and strengthening. For \textsc{TF-Var} with $T=x$ the result follows from the hypothesis $\Delta \vdash U$.
	\end{enumerate}
\end{proof}
\subsection{Unraveling}
By inspecting definition of the $\unravel$ function (\cref{def:unr}), we get the following immediate result.
\begin{proposition}
	\label{prop:unraveling}
	If $T=\unravel(T)$, then $T$ is one of $\tUnit$, $\tArrow{U}{m}{V}$, $\tRcdVrt{\ell}{T_\ell}{L}$, $x$, $\tMsg{U}$, $\tChoice{\ell}{S_\ell}{L}$, $\tSeq{\tMsg{U}}{S}$, $\tSkip$ or $\tEnd$. If $T \neq \unravel(T)$, then $T$ is one of $\tRec{x}{U}$, $\tSeq{\tChoice{\ell}{S_\ell}{L}}{S}$, $\tSeq{(\tSeq{S_1}{S_2})}{S_3}$, $\tSeq{\tSkip}{S}$, $\tSeq{\tEnd}{S}$ or $\tSeq{(\tRec{s}{S})}{R}$.
\end{proposition}
We can also define the notion of one-step unraveling for our types.
\begin{definition}
	We say that a type $T'$ is a \emph{one-step unraveling} of another type $T$, denoted $\unrOne(T)$, if: $T$ is a direct application of a type constructor, and $T'=T$; or $T$ is not a direct application of a type constructor, and $T'$ is obtained by one recursive call of the $\unravel$ function, which attempts to bring a type constructor into the front of a type.
\end{definition}
One example is $\unrOne(\tSeq{\tSkip}{S})=S$; another example is
$\unrOne(\tSeq{(\tSeq{S_1}{S_2})}{S_3}) = \tSeq{S_1}{(\tSeq{S_2}{S_3})}$. Notice that $T_0$ is contractive iff any sequence $T_0, T_1, ...$, where $T_{i+1}=\unrOne(T_i)$, eventually stabilises in $\unravel(T)$ (after finitely many steps).

Finally, we make some observations about the structure of subtyping derivations. We can classify the syntactic subtyping rules from \cref{fig:syntactic-sub} in three groups: \emph{progressing}, \emph{left-preserving} and \emph{right-preserving}:
\begin{itemize}
	\item We designate rules \textsc{S-Unit}, \textsc{S-Arrow}, \textsc{S-Rcd}, \textsc{S-Vrt}, \textsc{S-In}, \textsc{S-Out}, \textsc{S-IntChoice}, \textsc{S-ExtChoice}, \textsc{S-End}, \textsc{S-Skip}, \textsc{S-EndSeq1L}, \textsc{S-EndSeq1R}, \textsc{S-EndSeq2}, \textsc{S-InSeq1L}, \textsc{S-OutSeq1L}, \textsc{S-InSeq1R}, \textsc{S-OutSeq1R}, \textsc{S-InSeq2} and \textsc{S-OutSeq2} as \emph{progressing}. These rules consume the types on both sides of the relation, i.e., if we apply one of these rules from judgement $T\subS U$, we end up with judgements $T'\subS U'$ where $T',U'$ are both proper subterms of $T,U$. Moreover, they are applicable iff $T = \unravel(T)$ and $U = \unravel(U)$.
	
	\item We designate rules \textsc{S-RecL}, \textsc{S-SkipSeqL}, \textsc{S-ChoiceSeqL}, \textsc{S-SeqSeqL}, \textsc{S-RecSeqL} as \emph{right-preserving}. These rules change the type on the left-hand side of the relation, but preserve the type on the right-hand side. They are applicable when $T \neq \unravel(T)$.
	
	\item We designate rules \textsc{S-RecR}, \textsc{S-SkipSeqR}, \textsc{S-ChoiceSeqR}, \textsc{S-SeqSeqR}, \textsc{S-RecSeqR} are \emph{left-preserving}. These rules change the type on the right-hand side of the relation, but preseve the type on the left-hand side. They are applicable when $U \neq \unravel(U)$.
\end{itemize}
Furthermore, by inspecting the rules, we can gather that:
\begin{itemize}
	\item If we can apply a progressing rule for $T\subS U$, then it is the only applicable rule.
	\item If we can apply a left-preserving rule for $T\subS U$, then this the only left-preserving rule that can be applied (but a right-preserving rule may also be applicable).
	\item If we can apply a right-preserving rule for $T\subS U$, then this the only right-preserving rule that can be applied (but a left-preserving rule may also be applicable).
	\item If we can apply both a left-preserving rule and a right-preserving rule for $T \subS U$, then we can apply them one after the other in any order. Furthermore, both rules must eventually be applied in any successful derivation for $T \subS U$ .
\end{itemize}
From these observations we can immediately derive the following results.
\begin{lemma}
	\label{lemma:unraveling}
	\hfill
	\begin{enumerate}
		\item Let $T'=\unrOne(T)$ for some type $T$.
		\begin{enumerate}
			\item If $\Delta \vdash T$, then $\Delta \vdash T'$.
			\item $T\subS U$ iff $T'\subS U$.
			\item $\transition{T}{a}{U}$ iff $\transition{T'}{a}{U}$.
		\end{enumerate}
		\item Let $T'=\unravel(T)$ for some type $T$.
		\begin{enumerate}
			\item If $\Delta \vdash T$, then $\Delta \vdash T'$.
			\item $T\subS U$ iff $T'\subS U$.
			\item $\transition{T}{a}{U}$ iff $\transition{T'}{a}{U}$.
		\end{enumerate}
	\end{enumerate}
\end{lemma}
\begin{proof}
	Sub-item 1.a is immediate by inspection of the type formation rules. Sub-item 1.b follows from the preceding discussion. Sub-item 1.c is immediate by inspection of the LTS rules (\cref{fig:hocfst-lts}). Item 2 follows from Item 1 since $\unravel(T)$ is reached in a finite number of steps. 
\end{proof}
\section{Proof of Theorem \ref{thm:syntactic-preorder}}
In this section we focus on the proof of \cref{thm:syntactic-preorder}:
\syntacticPreorder*
\begin{proof}
	We need to prove reflexivity and transitivity.
	\paragraph*{Reflexivity} We prove by coinduction that $T\subS T$ for every type $T$ s.t. $\vdash T$. Consider the following relation.
	\begin{align*}
		\mathcal{R} = &\,\{(T,T) \svert \vdash T\}\\
		\cup &\,\{(T,T') \svert \vdash T \text{ and $T'=\unrOne(T)$}\}
	\end{align*}
	We shall prove that $\mathcal{R}$ is backward-closed for the rules of syntactic subtyping. This will show that $\mathcal{R}\subseteq{\subS}$ and, consequently, that $T\subS T$ for every type $T$.
	
	Let $(T,T)\in\mathcal{R}$. We consider first the cases in which $T$ fits a type constructor, i.e., $T = \unravel(T)$. Given that $T$ is well-formed, we have the following case analysis for it:
	
	(\textbf{Case} $T = \tUnit$): We apply axiom \textsc{S-Unit} to $(T,T)$.
	
	(\textbf{Case} $T = \tArrow{U}{m}{V}$): We apply rule \textsc{S-Arrow} to $(T,T)$, arriving at goals $(U,U)$ and $(V,V)$. Since the derivation of $\vdash T$ must use rule \textsc{TF-Arrow}, we also have that $\vdash U$ and $\vdash V$, and therefore that $(U,U),(V,V)\in\mathcal{R}$.
	
	(\textbf{Case} $T = \tRecord{\ell}{T_\ell}{L}$): We apply rule (\textsc{S-Rcd} and arrive at goals $(T_k,T_k)$ for each $k\in L$. The derivation of $\vdash T$ must use rule \textsc{TF-RcdVrt}, which implies that $\vdash T_k$ for each $k \in L$, which means that $(T_k, T_k)\in \mathcal{R}$ for each $k\in\mathcal{R}$.
	
	(\textbf{Case} $T = \tVariant{\ell}{T_\ell}{L}$): Analogous to case $T = \tRecord{\ell}{T_\ell}{L}$.
	
	(\textbf{Case} $T = x$): Cannot occur, for $\cancel{\vdash}x$. 
	
	(\textbf{Case} $T = \tEnd$): We apply axiom \textsc{S-End} to $(T,T)$.
	
	(\textbf{Case} $T = \tIn{U}$): We apply rule \textsc{S-In} to $(T,T)$, arriving at goal $(U,U)$. Since the derivation of $\vdash T$ must use rule \textsc{TF-Msg}, we have that $\vdash U$, and therefore that $(U,U)\in\mathcal{R}$.
	
	(\textbf{Case} $T = \tIntChoice{\ell}{T_\ell}{L}$): Analogous to case $T = \tRecord{\ell}{T_\ell}{L}$.
	
	(\textbf{Case} $T = \tExtChoice{\ell}{T_\ell}{L}$): Analogous to case $T = \tRecord{\ell}{T_\ell}{L}$.
	
	(\textbf{Case} $T = \tSkip$): We apply axiom \textsc{S-Skip} to $(T,T)$.
	
	(\textbf{Case} $T = \tSeq{\tIn{U}}{S}$): We apply rule \textsc{S-InSeq2}, arriving at goals $(U,U),(S,S)$. The derivation of $\vdash T$ must use rule \textsc{TF-Seq}, implying $\vdash \tIn{U}$ and $\Delta \vdash S$. Moreover, the derivation of $\vdash \tIn{U}$ must use rule \textsc{TF-Msg}, implying $\vdash U$. Since $\vdash U$ and $\vdash S$, we have $(U,U),(S,S)\in\mathcal{R}$. The case where $T=\tSeq{\tOut{U}}{S}$ is similar.
	
	(\textbf{Case} $T = \tSeq{s}{S}$): Cannot occur, since $\not\wellFormed{}{T}$. 
	
	Next, we consider cases in which $T \neq \unravel(T)$.
	
	(\textbf{Case} $T = \tRec{x}{U}$): We apply rule \textsc{S-RecR} to $(T,T)$, arriving at goal $(T,T')$ where $T'=\subst{x}{\tRec{x}{U}}{U}$. Since $T'=\unrOne(T)$, we have that $(T,T')\in\mathcal{R}$.
	
	(\textbf{Case} $T=\tSeq{\tEnd}{S}$): We apply axiom \textsc{S-EndSeq2}.
	
	(\textbf{Case} $T = \tSeq{\tChoice{\ell}{S_\ell}{L}}{R}$): We apply rule \textsc{S-ChoiceSeqR} to $(T,T)$, arriving at goal $(T,T')$ where $T' = \tChoice{\ell}{\tSeq{S_\ell}{R}}{L}$. Since $T'=\unrOne(T)$, we have that $(T,T')\in\mathcal{R}$.
	
	(\textbf{Case} $T = \tSeq{\tSkip}{S}$): We apply rule \textsc{S-SkipSeqR}, arriving at goal $(T,S)$. Since $S=\unrOne(T)$, we obtain that $(T,S)\in\mathcal{R}$.
	
	(\textbf{Case} $T = \tSeq{(\tSeq{S_1}{S_2})}{S_3}$): We apply rule \textsc{S-SeqSeqR}, arriving at goal $(T,T')$ where $T' = \tSeq{S_1}{(\tSeq{S_2}{S_3})}$. Since $T'=\unrOne(T)$, we obtain that $(T,T')\in\mathcal{R}$.
	
	\textbf{(Case $T = \tSeq{(\tRec{s}{S})}{R}$):} We apply rule \textsc{S-RecSeqR} to $(T,T)$, arriving at goal $(T,T')$, where $T'=\tSeq{(\subst{s}{\tRec{s}{S}}{S})}{R}$. Since $T'=\unrOne(T)$, we have that $(T,T')\in\mathcal{R}$.
	
	Next, we must consider cases $(T,T')\in\mathcal{R}$ where $T\neq T'$, which means, by definition, that $T'=\unrOne(T)$ and therefore that $T \neq \unravel(T)$. Given that $\vdash T$, we have the following case analysis for $T$.
	
	(\textbf{Case} $T = \tRec{x}{U}$): From \cref{lemma:type-substitution} follows that $\vdash \subst{x}{\tRec{x}{U}}{U}$. Since $T'=\unrOne(T)$, we know that $T'=\subst{x}{\tRec{x}{U}}{U}$. We apply rule \textsc{S-RecL} to $(T,T')$, arriving at goal $(T',T')\in\mathcal{R}$.
	
	(\textbf{Case} $T=\tSeq{\tEnd}{S}$): Then $T'=\tEnd$. We apply axiom \textsc{S-EndSeq1L} to $(T,T')$.
	
	(\textbf{Case} $T = \tSeq{\tChoice{\ell}{S_\ell}{L}}{R}$):  The derivation of $\vdash T$ must use rules \textsc{TF-Seq} and \textsc{TF-Choice}, implying that $\vdash S_k$ for each $k\in L$ and $\vdash U$. Again by rule \textsc{TF-Seq}, we get that $\vdash \tSeq{S_k}{R}$ for each $k\in L$ and thus, by rule \textsc{TF-Choice}, $\vdash \tChoice{\ell}{\tSeq{S_\ell}{R}}{L}$. Since $T'=\unrOne(T)$, we know that $T'=\tChoice{\ell}{\tSeq{S_\ell}{R}}{L}$. We apply rule \textsc{S-ChoiceSeqL} to $(T,T')$, arriving at goal $(T',T')\in\mathcal{R}$.
	
	(\textbf{Case} $T = \tSeq{\tSkip}{S}$): The derivation of $\vdash T$ must use rule \textsc{TF-Seq}, implying that $\vdash S$. Since $T'=\unrOne(T)$, we know that $T'= S$. We apply rule \textsc{E-SkipSeqL} to $(T,T')$, arriving at goal $(T',T')\in\mathcal{R}$.
	
	(\textbf{Case} $T = \tSeq{(\tSeq{S_1}{S_2})}{S_3}$): The derivation of $\vdash T$ must use rule \textsc{TF-Seq}, hence $\vdash S_1$, $\vdash S_2$, $\vdash S_3$. Therefore, by rule \textsc{TF-Seq} also, we have $\vdash \tSeq{S_1}{(\tSeq{S_2}{S_3})}$. Since $T'=\unrOne(T)$, we know that $T'= \tSeq{S_1}{(\tSeq{S_2}{S_3})}$. We apply rule \textsc{S-SeqSeqL} to $(T,T')$, arriving at goal $(T',T') \in \mathcal{R}$.
	
	(\textbf{Case} $T = \tSeq{(\tRec{s}{S})}{R}$): The derivation of $\vdash T$ must use rule \textsc{TF-Seq}, implying that $\vdash \tRec{s}{S}$ and $\vdash R$. From \cref{lemma:type-substitution} follows that $\vdash \subst{s}{\tRec{s}{S}}{S}$. By rule \textsc{TF-Seq} we get that $\vdash \tSeq{(\subst{s}{\tRec{s}{S}}{S})}{R}$. Since $T'=\unrOne(T)$, we know that $T'=\tSeq{(\subst{s}{\tRec{s}{S}}{S})}{R}$. We apply rule \textsc{S-RecSeqL} to $(T,T')$, arriving at goal $(T',T')\in\mathcal{R}$.
	
	\paragraph*{Transitivity} We now prove by coinduction that, for all types $T,U,V$ with $\vdash T, \vdash U, \vdash V$, if $T\subS U$ and $U\subS V$, then $T\subS V$. Consider the following relation.
	\begin{gather*}
		\mathcal{R} = \{(T,V)\svert\text{$\vdash T, \vdash V$ and there exists $U$ s.t. $\vdash U, T\subS U$ and $U\subS V$}\}
	\end{gather*}
	We prove that $\mathcal{R}$ is backward closed for the rules of the syntactic subtyping relation, showing that $\mathcal{R} \subseteq {\subS}$. This will give the desired property.
	
	We begin by assuming that no left-preserving rule applies to judgements $T \subS U$. Otherwise, we could apply its symmetric counterpart to judgement $U \subS V$ to get a different $U'$ for which $T \subS U'$ and $U' \subS V$. Without loss of generality, our derivation for $T \subS U$ starts with a finite sequence of left-preserving rules until we reach a type $U'$ that can only be consumed. We could then reach the same type $U'$ by applying a symmetric sequence of right-preserving rules to the derivation for $U \subS V$. We can therefore assume that $U$ is ready to be consumed.
	
	Suppose now a derivation for $T \subS U$ starting with a right-preserving rule, after which we get judgement $T' \subS U$ for some type $T'$. Here we can apply the corresponding rule to $(T,V)$, arriving at $(T',V)$, which is in $\mathcal{R}$ since $T' \subS U$ and $U \subS V$. The case where $U \subS V$ starts with a left-preserving rule can be handled similarly.
	
	What if both derivations for $T \subS U$ and $U \subS V$ start with a progressing rule? In this case, we need to inspect which rule is at the start of the derivation for $T \subS U$.
	
	(\textbf{Case} \textsc{S-Unit}): Then $T=\tUnit$ and $U=\tUnit$. The only progressing rule that can be applied at $U \subS V$ is also \textsc{S-Unit}, implying that $V=\tUnit$ as well. Therefore, we can apply axiom \textsc{S-Unit} to $(T,V)$.
	
	(\textbf{Case} \textsc{S-Arrow}): Then $T = \tArrow{T_1}{m}{T_2}$ and $U = \tArrow{U_1}{n}{U_2}$ for some $T_1,T_2,U_1,U_2,m,n$. Furthermore, we have $U_1\subS T_1$ and $T_2 \subS U_2$ and $m \sqsubseteq n$. The only progressing rule that can be applied to $U \subS V$ is also \textsc{S-Arrow}, implying that $V = \tArrow{V_1}{o}{V_2}$ for some $V_1,V_2,o$. Furthermore, we have $V_1 \subS U_1$, $U_2 \subS V_2$ and $n \sqsubseteq o$. By transitivity of ${\sqsubseteq}$ we obtain $m \sqsubseteq o$. We apply rule \textsc{S-Arrow} to $(T,V)$, arriving at goals $(V_1,T_1),(T_2,V_2)\in\mathcal{R}$.
	
	(\textbf{Case} \textsc{S-Rcd}): Then $T=\tRecord{\ell}{T_\ell}{L}$ and $U=\tRecord{k}{U_k}{K}$ for some $L,K,T_i,U_j,i\in L, j\in K$. Furthermore, we have $K\subseteq L$, $T_j \subS U_j$ for $j \in K$. The only progressing rule that can be applied to $U \subS V$ is also \textsc{S-Rcd}, implying that $V = \tRecord{h}{V_h}{H}$ for some $H,V_h,h\in H$. Furthermore, we have $H \subseteq K$ $U_h \subS V_h$ for each $h \in H$. By transitivity of ${\subseteq}$ we get $H \subseteq L$. We apply rule \textsc{S-Rcd} to $(T,V)$, arriving at goals $(T_h, V_h)\in\mathcal{R}$ for each $h\in H$. Case \textsc{S-IntChoice} is similar.
	
	(\textbf{Case} \textsc{S-Vrt}): Then $T=\tVariant{\ell}{T_\ell}{L}$ and $U=\tVariant{k}{U_k}{K}$ for some $L,K,T_i,U_j,i\in L, j\in K$. Furthermore, we have $L\subseteq K$, $T_j \subS U_j$ for $j \in K$. The only progressing rule that can be applied to $U \subS V$ is also \textsc{S-Vrt}, implying that $V = \tVariant{h}{V_h}{H}$ for some $H,V_h,h\in H$. Furthermore, we have $K \subseteq H$, $U_h \subS V_h$ for each $h \in H$. By transitivity of ${\subseteq}$ we get $L \subseteq H$. We apply rule \textsc{S-Vrt} to $(T,V)$, arriving at goals $(T_h, V_h)\in\mathcal{R}$ for each $h\in H$. Case \textsc{S-ExtChoice} is similar.
	
	(\textbf{Case} \textsc{S-End}): Then $T=\tEnd$ and $U = \tEnd$. The two possible progressing rules for $U\subS V$ are \textsc{S-End} and \textsc{S-EndSeq1R}. In the first case we have $V=\tEnd$, so we apply \textsc{S-End} to $(T,V)$. In the second case we have $V=\tSeq{\tEnd}{S}$ for some $S$, so we apply rule \textsc{S-EndSeq1R} to $(T,V)$.
	
	(\textbf{Case} \textsc{S-In}): Then $T=\tIn{T'}$ and $U=\tIn{U'}$ for some $T',U'$. Furthermore, we have $T' \subS U'$. The two possible progressing rules for $U\subS V$ are \textsc{S-In} and \textsc{S-InSeq1R}. In the first case, we have $V=\tIn{V'}$ for some $V'$. It follows that $U' \subS V'$. We apply rule \textsc{S-In} to $(T,V)$, arriving at goal $(T',V')\in\mathcal{R}$. In the second case, we have $V = \tSeq{\tIn{V'}}{S}$ for some $V'$ and $S$. It follows that $U' \subS V'$ and $S \subS \tSkip$. We apply rule \textsc{S-InSeq1R} to $(T,V)$, arriving at goal $(T',V')\in\mathcal{R}$. The cases \textsc{S-InSeq1L}, \textsc{S-InSeq1R}, \textsc{S-InSeq2} are handled similarly. 
	
	(\textbf{Case} \textsc{S-Out}): Then $T=\tOut{T'}$ and $U=\tOut{U'}$ for some $T',U'$. Furthermore, we have $U' \subS T'$. The two possible progressing rules for $U\subS V$ are \textsc{S-Out} and \textsc{S-OutSeq1R}. In the first case, we have $V=\tOut{V'}$ for some $V'$. It follows that $V' \subS U'$. We apply rule \textsc{S-Out} to $(T,V)$, arriving at goal $(V',T')\in\mathcal{R}$. In the second case, we have $V = \tSeq{\tOut{V'}}{S}$ for some $V'$ and $S$. It follows that $V' \subS U'$ and $S\subS \tSkip$. We apply rule \textsc{S-OutSeq1R} to $(T,V)$, arriving at goal $(V',T')\in\mathcal{R}$. The cases \textsc{S-OutSeq1L}, \textsc{S-OutSeq1R}, \textsc{S-OutSeq2} are handled similarly. 
	
	(\textbf{Case} \textsc{S-EndSeq1L}): Then $T = \tSeq{\tEnd}{S}$ and $U=\tEnd$ for some $S$. The two possible progressing rules for $U \subS V$ are \textsc{S-End} and \textsc{S-EndSeq1R}. In the first case we have $V=\tEnd$, so we can apply rule \textsc{S-EndSeq1L} to $(T,V)$. In the second case we have $V=\tSeq{\tEnd}{R}$ for some $R$, so we can apply rule \textsc{S-EndSeq2} to $(T,V)$.
	
	(\textbf{Case} \textsc{S-EndSeq1R}): Then $T=\tEnd$ and $U=\tSeq{\tEnd}{S}$ for some $S$. The two possible progressing rules for $U \subS V$ are \textsc{S-EndSeq1L} and \textsc{S-EndSeq2}. In the first case we have $V = \tEnd$, so we can apply \textsc{S-End} to $(T,V)$. In the second case we have $V = \tSeq{\tEnd}{R}$ for some $R$, so we can apply \textsc{S-EndSeq1R} to $(T,V)$.
	
	(\textbf{Case} \textsc{S-EndSeq2}): Then $T=\tSeq{\tEnd}{S}$ and $U=\tSeq{\tEnd}{R}$ for some $S,R$. The two possible progressing rules for $U \subS V$ are \textsc{S-EndSeq1L} and \textsc{S-EndSeq2}. In the first case we have $V=\tEnd$, so we can apply \textsc{S-EndSeq1L} to $(T,V)$. In the second case we have $V=\tSeq{\tEnd}{S'}$ for some $S'$, so we apply \textsc{S-EndSeq2} to $(T,V)$.
\end{proof}

\section{Proof of Theorem \ref{thm:soundness-completeness-stx-sem}} \label{sec:proof-soundness-completeness-stx-sem}

Recall the statement of \cref{thm:soundness-completeness-stx-sem}, presented in \cref{sec:semantic-subtyping}.
\syntacticSemanticSoundnessCompleteness*
\begin{proof}
	We analyse both directions of the biconditional separately.
	
	\paragraph*{Direct implication}
	Consider the relation $\mathcal{R} = \{(T,U)\svert \text{$\vdash T$, $\vdash U$ and $T\subS U$}\}$. 
	We must show that $\mathcal{R}$ is an $\mathcal{XYZW}$-simulation with $\mathcal{X,Y,Z,W}$ as defined in \cref{def:hocfst-sim}. This will show that $\mathcal{R}\subseteq{\subS}$, and hence that $T\subS U$ implies $T\simS U$.
	
	The proof has two parts. First consider cases $(T,U)\in\mathcal{R}$ s.t. $\unravel(T)=T$ and $\unravel(U)=U$. We proceed by case analysis for the last rule in the derivation of $T \subS U$, which must be a progressing rule.
	
	(\textbf{Case} \textsc{S-Unit}): Then $T=\tUnit$ and $U=\tUnit$. The unique transition that can be applied to $\tUnit$ is $\transition{\tUnit}{\tUnit}{\tSkip}$ (\textsc{L-Unit}). Since $\tUnit \in \simSX,\simSY$, we should have that if $\transition{T}{\tUnit}{T'}$ for some $T'$, then  $\transition{U}{\tUnit}{U'}$ for some $U'$ with $(T',U')\in\mathcal{R}$, and also that if $\transition{U}{\tUnit}{U'}$ for some $U'$, then  $\transition{T}{\tUnit}{T'}$ for some $T'$ with $(T',U')\in\mathcal{R}$. It is readily verifiable that the single transitions of both $T$ and $U$ match each other. That $(\tSkip,\tSkip)\in\mathcal{R}$ follows from \textsc{S-Skip} and the definition of $\mathcal{R}$.
	
	(\textbf{Case} \textsc{S-Arrow}): Then $T=\tArrow{T_1}{m}{T_2}$, $U=\tArrow{U_1}{n}{U_2}$, $U_1\subS T_1$, $T_2\subS U_2$ and $m\sqsubseteq n$. The only transitions that can be applied to $T$ are $\transition{T}{\lArrowDomain}{T_1}$, $\transition{T}{\lArrowRange}{T_2}$ and, if $m={\lin}$, $\transition{T}{\lArrowLinear}{\tSkip}$ (\textsc{L-ArrowDom}, \textsc{L-ArrowRng} and \textsc{L-LinArrow}). Similarly, the only transitions applicable to $U$ are $\transition{U}{\lArrowDomain}{U_1}$, $\transition{U}{\lArrowRange}{U_2}$ and, if $n={\lin}$, $\transition{U}{\lArrowLinear}{\tSkip}$. 
	
	Since $\lArrowDomain\in \simSZ,\simSW$, we should have that if $\transition{T}{\lArrowDomain}{T'}$ for some $T'$, then  $\transition{U}{\lArrowDomain}{U'}$ for some $U'$ with $(U',T')\in\mathcal{R}$, and also that if $\transition{U}{\lArrowDomain}{U'}$ for some $U'$, then  $\transition{T}{\lArrowDomain}{T'}$ for some $T'$ with $(U',T')\in\mathcal{R}$. That the transitions of $T$ and $U$ by $\lArrowDomain$ match is readily verifiable, and that $(U_1,T_1)\in\mathcal{R}$ is given by $U_1 \subS T_1$ and the definition of $\mathcal{R}$. 
	
	Similarly, since $\lArrowRange\in \simSX,\simSY$, we should have that if $\transition{T}{\lArrowRange}{T'}$ for some $T'$, then  $\transition{U}{\lArrowRange}{U'}$ for some $U'$ with $(T',U')\in\mathcal{R}$, and also that if $\transition{U}{\lArrowRange}{U'}$ for some $U'$, then  $\transition{T}{\lArrowRange}{T'}$ for some $T'$ with $(T',U')\in\mathcal{R}$. That the transitions of $T$ and $U$ by $\lArrowRange$ match is readily verifiable, and that $(T_2,U_2)\in\mathcal{R}$ is given by $T_2 \subS U_2$ and the definition of $\mathcal{R}$. 
	
	Finally, as $\lArrowLinear\in\simSX$, we need to show that if $\transition{T}{\lArrowLinear}{T'}$ for some $T'$, then $\transition{U}{\lArrowLinear}{U'}$ for some $U'$ with $(T',U')\in\mathcal{R}$. As we have seen, $T$ can have at most one such transition, $\transition{T}{\lArrowLinear}{\tSkip}$, and only in the case where $m={\lin}$. From $m\sqsubseteq n$ follows that $n={\lin}$. From this follows that $\transition{U}{\lArrowLinear}{\tSkip}$. We arrive at pair $(\tSkip, \tSkip)$, which is obviously in $\mathcal{R}$.
	
	(\textbf{Case} \textsc{S-Rcd}): Then $T=\tRecord{\ell}{T_\ell}{L}$, $U=\tRecord{k}{U_k}{K}$ and $K\subseteq L$. The only transitions that can be applied to $T$ are $\transition{T}{\lRcdDefault}{\tSkip}$ and $\transition{T}{\lRcdLabel{i}}{T_i}$ for each $i\in L$ (\textsc{L-RcdVrt} and \textsc{L-RcdVrtField}). Similarly, the only transitions that can be applied to $U$ are $\transition{U}{\lRcdDefault}{\tSkip}$ and $\transition{U}{\lRcdLabel{j}}{U_j}$ for each $j\in K$. It is clear from the rules of type formation that $\vdash T_i$ and $\vdash U_j$ for each $i\in L, j\in K$ and, because \textsc{S-Rcd} was used, we know $T_j \subS U_j$ for each $j \in K$.
	
	Since $\lRcdDefault\in\simSX,\simSY$, we should have that if $\transition{T}{\lRcdDefault}{T'}$ for some $T'$, then  $\transition{U}{\lRcdDefault}{U'}$ for some $U'$ with $(T',U')\in\mathcal{R}$, and also that if $\transition{U}{\lRcdDefault}{U'}$ for some $U'$, then  $\transition{T}{\lRcdDefault}{T'}$ for some $T'$ with $(T',U')\in\mathcal{R}$. That the transitions of $T$ and $U$ by $\lRcdDefault$ match is readily verifiable, and that $(\tSkip,\tSkip)\in\mathcal{R}$ is also evident.
	
	Finally, since $\lRcdLabel{j}\in\simSY$ for each $j\in K$, we need to show that if $\transition{U}{\lRcdLabel{j}}{U'}$ for some $U'$, then $\transition{T}{\lRcdLabel{j}}{T'}$ for some $T'$ with $(T',U')\in\mathcal{R}$. As $K\subseteq L$, it is readily verifiable that $T$ matches every transition of $U$ by $\lRcdLabel{j}$ for each $j\in K$. That $(T_j,U_j)\in\mathcal{R}$ follows by $T_j \subS U_j$ for each $j \in K$ and the definition of $\mathcal{R}$.
	
	(\textbf{Case} \textsc{S-Vrt}): Then $T=\tVariant{\ell}{T_\ell}{L}$, $U=\tVariant{k}{U_k}{K}$ and $L\subseteq K$. The only transitions that can be applied to $T$ are $\transition{T}{\lVrtDefault}{\tSkip}$ and $\transition{T}{\lVrtDefault{i}}{T_i}$ for each $i\in L$ (\textsc{L-RcdVrt} and \textsc{L-RcdVrtField}). Similarly, the only transitions that can be applied to $U$ are $\transition{U}{\lVrtDefault}{\tSkip}$ and $\transition{U}{\lVrtLabel{j}}{U_j}$ for each $j\in K$. It is clear from the rules of type formation that $\vdash T_i$ and $\vdash U_j$ for each $i\in L, j\in K$ and, because \textsc{S-Vrt} was used, we know $T_i \subS U_i$ for each $i \in L$.
	
	Since $\lVrtDefault\in\simSX,\simSY$, we should have that if $\transition{T}{\lVrtDefault}{T'}$ for some $T'$, then  $\transition{U}{\lVrtDefault}{U'}$ for some $U'$ with $(T',U')\in\mathcal{R}$, and also that if $\transition{U}{\lVrtDefault}{U'}$ for some $U'$, then  $\transition{T}{\lVrtDefault}{T'}$ for some $T'$ with $(T',U')\in\mathcal{R}$. That the transitions of $T$ and $U$ by $\lVrtDefault$ match is readily verifiable, and that $(\tSkip,\tSkip)\in\mathcal{R}$ is also evident.
	
	Finally, since $\lVrtLabel{i}\in\simSX$ for each $i\in L$, we need to show that if $\transition{T}{\lVrtLabel{i}}{T'}$ for some $T'$, then $\transition{U}{\lVrtLabel{i}}{U'}$ for some $U'$ with $(T',U')\in\mathcal{R}$. As $L\subseteq K$, it is readily verifiable that $U$ matches every transition of $T$ by $\lVrtLabel{i}$ for each $i\in L$. That $(T_i,U_i)\in\mathcal{R}$ follows by $T_i \subS U_i$ for each $i \in L$ and the definition of $\mathcal{R}$.
	
	(\textbf{Case} \textsc{S-End}): Analogous to case \textsc{S-Unit}.
	
	(\textbf{Case} \textsc{S-In}): Then $T=\tIn{T'}$ and $U=\tIn{U'}$. The only transitions that can be applied to $T$ are $\transition{T}{\lInPayload}{T'}$ and $\transition{T}{\lInContinuation}{\tSkip}$ (\textsc{L-Msg1}, \textsc{L-Msg2}). Similarly, the only transitions that can be applied to $U$ are $\transition{U}{\lInPayload}{U'}$ and $\transition{U}{\lInContinuation}{\tSkip}$. It is clear from the rules of type formation that $\vdash T'$ and $\vdash U'$. Furthermore, because \textsc{S-In} was used, $T'\subS U'$. 
	
	Since $\lInPayload,\lInContinuation\in\simSX,\simSY$, we should have that if $\transition{T}{\lInPayload}{T'}$ for some $T'$, then  $\transition{U}{\lInPayload}{U'}$ for some $U'$ with $(T',U')\in\mathcal{R}$, and similarly for $\lInContinuation$. We must also have that if $\transition{U}{\lInPayload}{U'}$ for some $U'$, then  $\transition{T}{\lInPayload}{T'}$ for some $T'$ with $(T',U')\in\mathcal{R}$, and similarly for $\lInContinuation$. That the transitions of $T$ and $U$ by $\lInPayload$ and $\lInContinuation$ match is readily verifiable, and $(T',U'),(\tSkip,\tSkip)\in\mathcal{R}$ follows from $T'\subS U'$, \textsc{S-Skip} and from the definition of $\mathcal{R}$.
	
	(\textbf{Case} \textsc{S-Out}): Then $T=\tOut{T'}$ and $U=\tOut{U'}$. The only transitions that can be applied to $T$ are $\transition{T}{\lOutPayload}{T'}$ and $\transition{T}{\lOutContinuation}{\tSkip}$ (\textsc{L-Msg1}, \textsc{L-Msg2}). Similarly, the only transitions that can be applied to $U$ are $\transition{U}{\lOutPayload}{U'}$ and $\transition{U}{\lOutContinuation}{\tSkip}$. It is clear from the rules of type formation that $\vdash T'$ and $\vdash U'$. Furthermore, because \textsc{S-Out} was used, $U'\subS T'$. 
	
	Since $\lOutPayload\in \simSZ,\simSW$, we should have that if $\transition{T}{\lOutPayload}{T''}$ for some $T''$, then  $\transition{U}{\lOutPayload}{U''}$ for some $U''$ with $(U'',T'')\in\mathcal{R}$, and also that if $\transition{U}{\lOutPayload}{U''}$ for some $U''$, then  $\transition{T}{\lOutPayload}{T''}$ for some $T''$ with $(U'',T'')\in\mathcal{R}$. That the transitions of $T$ and $U$ by $\lOutPayload$ match is readily verifiable, and that $(U',T')\in\mathcal{R}$ is given by $U' \subS T'$ and the definition of $\mathcal{R}$.
	
	Finally, since $\lOutContinuation\in\simSX,\simSY$, we should have that if $\transition{T}{\lOutContinuation}{T'}$ for some $T'$, then  $\transition{U}{\lOutContinuation}{U'}$ for some $U'$ with $(T',U')\in\mathcal{R}$, and also that if $\transition{U}{\lOutContinuation}{U'}$ for some $U'$, then  $\transition{T}{\lOutContinuation}{T'}$ for some $T'$ with $(T',U')\in\mathcal{R}$. That the transitions of $T$ and $U$ by $\lOutContinuation$ match is readily verifiable, and that $(\tSkip,\tSkip)\in\mathcal{R}$ follows from \textsc{S-Skip} and the definition of $\mathcal{R}$.
	
	(\textbf{Case} \textsc{S-ExtChoice}): Analogous to case \textsc{S-Vrt}.
	
	(\textbf{Case} \textsc{S-IntChoice}): Analogous to case \textsc{S-Rcd}.
	
	(\textbf{Case} \textsc{S-Skip}): Then $T=\tSkip$ and $U=\tSkip$. Since no transitions apply to $\tSkip$, the conditions for $\mathcal{XYZW}$-simulation trivially hold.
	
	(\textbf{Case} \textsc{S-EndSeq1L}): Then $T=\tSeq{\tEnd}{S}$ and $U=\tEnd$. The only transition that can be applied to $T$ is $\transition{T}{\tEnd}{\tSkip}$ (\textsc{L-EndSeq}). Similarly, the only transition that can be applied to $U$ is $\transition{U}{\tEnd}{\tSkip}$ (\textsc{L-End}).
	
	Since $\tEnd\in\simSX,\simSY$, we should have that if $\transition{T}{\tEnd}{T'}$ for some $T'$, then  $\transition{U}{\tEnd}{U'}$ for some $U'$ with $(T',U')\in\mathcal{R}$, and also that if $\transition{U}{\tEnd}{U'}$ for some $U'$, then  $\transition{T}{\tEnd}{T'}$ for some $T'$ with $(T',U')\in\mathcal{R}$. That the transitions of $T$ and $U$ by $\tEnd$ match is readily verifiable, and $(\tSkip,\tSkip)\in\mathcal{R}$ follows from \textsc{S-Skip} and the definition of $\mathcal{R}$.
	
	(\textbf{Case} \textsc{S-EndSeq1R}, \textsc{S-EndSeq2}): Analogous to case \textsc{S-EndSeq1L}.
	
	(\textbf{Case} \textsc{S-InSeq1L}): Then $T=\tSeq{\tIn{T'}}{S}$ and $U=\tIn{U'}$. The only transitions that can be applied to $T$ are $\transition{T}{\lInPayload}{T'}$ and $\transition{T}{\lInContinuation}{S}$ (\textsc{L-MsgSeq1}, \textsc{MsgSeq2}). Similarly, the only transitions that can be applied to $U$ are $\transition{U}{\lInPayload}{U'}$ and $\transition{U}{\lInContinuation}{\tSkip}$. It is clear from the rules of type formation that $\vdash T'$,$\vdash U'$ and $\vdash S$. Furthermore, because \textsc{S-InSeq1L} was used, $T'\subS U'$ and $S\subS \tSkip$. 
	
	Since $\lInPayload,\lInContinuation\in\simSX,\simSY$, we should have that if $\transition{T}{\lInPayload}{T'}$ for some $T'$, then  $\transition{U}{\lInPayload}{U'}$ for some $U'$ with $(T',U')\in\mathcal{R}$, and similarly for $\lInContinuation$. For the same reason, we must also have that if $\transition{U}{\lInPayload}{U'}$ for some $U'$, then  $\transition{T}{\lInPayload}{T'}$ for some $T'$ with $(T',U')\in\mathcal{R}$, and similarly for $\lInContinuation$. That the transitions of $T$ and $U$ by $\lInPayload$ and $\lInContinuation$ match is readily verifiable, and $(T',U'),(S,\tSkip)\in\mathcal{R}$ follows from $T'\subS U'$, from $S \subS \tSkip$ and from the definition of $\mathcal{R}$.
	
	(\textbf{Case} \textsc{S-InSeq1R}, \textsc{S-InSeq2}): Analogous to case \textsc{S-InSeq1L}.
	
	(\textbf{Case} \textsc{S-OutSeq1L}): Then $T=\tSeq{\tOut{T'}}{S}$ and $U=\tOut{U'}$. The only transitions that can be applied to $T$ are $\transition{T}{\lOutPayload}{T'}$ and $\transition{T}{\lOutContinuation}{S}$ (\textsc{L-MsgSeq1}, \textsc{MsgSeq2}). Similarly, the only transitions that can be applied to $U$ are $\transition{U}{\lOutPayload}{U'}$ and $\transition{U}{\lOutContinuation}{\tSkip}$. It is clear from the rules of type formation that $\vdash T'$, $\vdash U'$ and $\vdash S$. Furthermore, because \textsc{S-OutSeq1L} was used, $U'\subS T'$ and $S\subS \tSkip$. 
	
	Since $\lOutPayload\in \simSZ,\simSW$, we should have that if $\transition{T}{\lOutPayload}{T''}$ for some $T''$, then  $\transition{U}{\lOutPayload}{U''}$ for some $U''$ with $(U'',T'')\in\mathcal{R}$, and also that if $\transition{U}{\lOutPayload}{U''}$ for some $U''$, then  $\transition{T}{\lOutPayload}{T''}$ for some $T''$ with $(U'',T'')\in\mathcal{R}$. That the transitions of $T$ and $U$ by $\lOutPayload$ match is readily verifiable, and that $(U',T')\in\mathcal{R}$ is given by $U' \subS T'$ and the definition of $\mathcal{R}$.
	
	Finally, since $\lOutContinuation\in\simSX,\simSY$, we should have that if $\transition{T}{\lOutContinuation}{T'}$ for some $T'$, then  $\transition{U}{\lOutContinuation}{U'}$ for some $U'$ with $(T',U')\in\mathcal{R}$, and also that if $\transition{U}{\lOutContinuation}{U'}$ for some $U'$, then  $\transition{T}{\lOutContinuation}{T'}$ for some $T'$ with $(T',U')\in\mathcal{R}$. That the transitions of $T$ and $U$ by $\lOutContinuation$ match is readily verifiable, and $(S,\tSkip)\in\mathcal{R}$ follows from $S \subS \tSkip$ and the definition of $\mathcal{R}$.
	
	(\textbf{Case} \textsc{S-OutSeq1R}, \textsc{S-OutSeq2}): Analogous to case \textsc{S-OutSeq1L}.
	
	Now consider that $T \neq \unravel(T)$. From $T\subS U$ follows that $\unravel(T)\subS U$, resulting from  the application of right-preserving rules. If $U=\unravel(U)$, then the above case analysis shows that the conditions for $\mathcal{XYZW}$-simulation between $\unravel(T)$ and $U$ hold. Since $T$ and $\unravel(T)$ have the same transitions, i.e., $\transition{T}{a}{T'}$ iff $\transition{\unravel(T)}{a}{T'}$ for some $T'$, the conditions for $\mathcal{XYZW}$-simulation between $T$ and $U$ also hold. Otherwise, if $U\neq\unravel(U)$, it similarly follows that $\unravel(T)\subS\unravel(U)$, resulting from the application of left-preserving rules. The previous case analysis shows that the conditions for $\mathcal{XYZW}$-simulation between $\unravel(T)$ and $\unravel(U)$ hold. Since $U$ and $\unravel(U)$ have the same transitions, the same conditions also hold between $T$ and $U$. The case with $T=\unravel(U), U\neq \unravel(U)$ is analogous.\smallskip
	\paragraph*{Reverse implication}
	Consider the relation $\mathcal{S} = \{(T,U)\svert\text{$\wellFormed{}{T}$, $\wellFormed{}{U}$ and $T \simS U$}\}$.
	We prove that relation $\mathcal{S}$ is backward closed for the rules of the syntactic subtyping relation. This will show that $\mathcal{S}\subseteq {\subS}$, and hence that $T\simS U$ implies $T \subS U$.
	
	The proof has two parts. First, consider the cases where both $T$ and $U$ fit a type constructor, i.e., $T=\unravel(T)$ and $U=\unravel(U)$. We proceed by case analysis on the structure of $T$.
	
	(\textbf{Case} $T=\tUnit$): The only transition that applies to $T$ is $\transition{T}{\tUnit}{\tSkip}$. Since $T\simS U$ and $U=\unravel(U)$, then $U=\tUnit$. Therefore we can apply \textsc{S-Unit}.
	
	(\textbf{Case} $T=\tArrow{T_1}{m}{T_2}$): Two transitions apply to $T$ regardless of $m$: $\transition{T}{\lArrowDomain}{T_1}$ and $\transition{T}{\lArrowRange}{T_2}$. Since $T \simS U$ and $U = \unravel(U)$, we know that $U = \tArrow{U_1}{n}{U_2}$ and that, regardless of $n$, $\transition{U}{\lArrowDomain}{U_1}$ and $\transition{U}{\lArrowRange}{U_2}$. Furthermore, we know that $U_1 \simS T_1$ (since ${\lArrowDomain}\in\simSZ,\simSW$) and that  $T_2 \simS U_2$ (since ${\lArrowRange} \in \simSX,\simSY$).
	
	Before we apply \textsc{S-Arrow}, we need to have $m\sqsubseteq n$. We know that $\transition{T}{\lArrowLinear}{\tSkip}$ iff $m={\lin}$ and that $\transition{U}{\lArrowLinear}{\tSkip}$ iff $n={\lin}$. Recall that ${\lArrowLinear}\in\simSX$, which means that the only case where $n\not\sqsubseteq m$ ($m=1$ and $n=*$) cannot occur, for it would contradict $T \simS U$ (since $U$ could not match a transition of $T$ by a label in $\simSX$). 
	
	We can therefore apply \textsc{S-Arrow}, arriving at $(U_1,T_1),(T_2,U_2)\in\mathcal{S}$. 
	
	(\textbf{Case} $T=\tRecord{\ell}{T_\ell}{L}$): The only transitions that can be applied to $T$ are $\transition{T}{\lRcdDefault}{\tSkip}$ and $\transition{T}{\lRcdLabel{i}}{T_i}$ for each $i\in L$. Since $T \simS U$ and $U=\unravel(U)$, we have $U = \tRecord{k}{U_k}{K}$ with transitions $\transition{U}{\lRcdDefault}{\tSkip}$ and $\transition{U}{\lRcdLabel{j}}{U_j}$ for $j \in K$. Since labels of the form $\lRcdLabel{\ell}$ belong to $\simSY$, we know that $K\subseteq L$, for $T$ must be able to match all transitions of $U$ by $\lRcdLabel{j}$ for each $j\in K$. From this we obtain $T_j \simS U_j$ for each $j\in K$, arriving at $(T_j,U_j)\in\mathcal{S}$ for each $j\in K$. 
	
	(\textbf{Case} $T=\tVariant{\ell}{T_\ell}{L}$): The only transitions that can be applied to $T$ are $\transition{T}{\lVrtDefault}{\tSkip}$ and $\transition{T}{\lVrtLabel{i}}{T_i}$ for each $i\in L$. Since $T \simS U$ and $U=\unravel(U)$, we have $U = \tVariant{k}{U_k}{K}$ with transitions $\transition{U}{\lVrtDefault}{\tSkip}$ and $\transition{U}{\lVrtLabel{j}}{U_j}$ for $j \in K$. Since labels of the form $\lVrtLabel{\ell}$ belong to $\simSX$, we know that $L\subseteq K$, for $U$ must be able to match all transitions of $T$ by $\lVrtLabel{i}$ for each $i\in L$. From this we obtain $T_i \simS U_i$ for each $i\in L$, arriving at $(T_i,U_i)\in\mathcal{S}$ for each $i\in L$. 
	
	(\textbf{Case} $T=x$): Cannot occur, since $\not\wellFormed{}{T}$.
	
	(\textbf{Case} $T=\tEnd$): Analogous to the case where $T=\tUnit$. 
	
	(\textbf{Case} $T=\tIn{T'}$): The only transitions applicable to $T$ are $\transition{T}{\lInPayload}{T'}$ and $\transition{T}{\lInContinuation}{\tSkip}$. Since $T \simS U$ and $U = \unravel(U)$, then either $U=\tIn{U'}$ or $U=\tSeq{\tIn{U'}}{V}$ with $V \simS \tSkip$. In either case, $T' \simS U'$. In the first case, we can apply \textsc{S-In}, arriving at $(T',U')\in\mathcal{S}$. In the second case, we can apply \textsc{S-InSeq1R}, arriving at $(T',U'), (V,\tSkip)\in\mathcal{S}$.
	
	(\textbf{Case} $T=\tOut{T'}$): The only transitions applicable to $T$ are $\transition{T}{\lOutPayload}{T'}$ and $\transition{T}{\lOutContinuation}{\tSkip}$. Since $T \simS U$ and $U = \unravel(U)$, then either $U=\tOut{U'}$ or $U=\tSeq{\tOut{U'}}{V}$ with $V \simS \tSkip$. In either case, $U' \simS T'$. In the first case, we can apply \textsc{S-Out}, arriving at $(U',T')\in\mathcal{S}$. In the second case we can apply \textsc{S-OutSeq1R}, arriving at $(T',U'),(V,\tSkip)\in\mathcal{S}$.
	
	(\textbf{Case} $T=\tIntChoice{\ell}{S_\ell}{L}$): Analogous to case $T=\tRecord{\ell}{T_\ell}{L}$.
	
	(\textbf{Case} $T=\tExtChoice{\ell}{S_\ell}{L}$): Analogous to case $T=\tVariant{\ell}{T_\ell}{L}$.
	
	(\textbf{Case} $T=\tSkip$): No transitions apply to $T$. Since $T\simS U$ and $U=\unravel(U)$, then $U=\tSkip$. Therefore we can apply \textsc{S-Skip}.
	
	(\textbf{Case} $T=\tSeq{\tIn{T_1}}{T_2}$): The only transitions applicable to $T$ are $\transition{T}{\lInPayload}{T_1}$ and $\transition{T}{\lInContinuation}{T_2}$. Since $T \simS U$ and $U = \unravel(U)$, then either $U=\tIn{U_1}$ or $U=\tSeq{\tIn{U_1}}{U_2}$. In the first case, $T_1 \simS U_1$ and $T_2 \simS \tSkip$; we can apply \textsc{S-InSeq1L}, arriving at $(T_1,U_1)\in\mathcal{S}$. In the second case, $T_1 \simS U_1$ and $T_2 \simS U_2$; we can apply \textsc{S-InSeq2}, arriving at $(T_1,U_1),(T_2,U_2)\in\mathcal{S}$.
	
	(\textbf{Case} $T=\tSeq{\tOut{T_1}}{T_2}$): The only transitions applicable to $T$ are $\transition{T}{\lOutPayload}{T_1}$ and $\transition{T}{\lOutContinuation}{T_2}$. Since $T \simS U$ and $U = \unravel(U)$, then either $U=\tOut{U_1}$ or $U=\tOut{\tIn{U_1}}{U_2}$. In the first case, $U_1 \simS T_1$ and $T_2 \simS \tSkip$; we can apply \textsc{S-OutSeq1L}, arriving at $(U_1,T_1)\in\mathcal{S}$. In the second case, $U_1 \simS T_1$ and $T_2 \simS U_2$; we can apply \textsc{S-OutSeq2}, arriving at $(U_1,T_1),(T_2,U_2)\in\mathcal{S}$.
	
	Now consider that $T \neq \unravel(T)$. From $T \simS U$ it follows that $T' \simS U$ where $T'=\unrOne(T)$, due to the fact that $T$ and $T'$ have the same transitions (\cref{lemma:unraveling}). Then we can apply an appropriate right-preserving rule to $(T,U)$, arriving at $(T',U)\in\mathcal{S}$. The case with $T=\unravel(T),U\neq\unravel(U)$ is analogous.
\end{proof}
\section{Proof of Theorem \ref{thm:soundness-grammars}} \label{sec:proof-soundness-grammars}
\begin{lemma}
	\label{lemma:bisim-implies-sim}
	Let $\vec X$ and $\vec Y$ be words from a simple GNF grammar with productions $\prods$. If $\vec X \sim_\prods \vec Y$, then $\vec X \simS_\prods \vec Y$ and $\vec Y \simS_\prods \vec X$. 
\end{lemma}
\begin{proof}
	Consider the relation $\mathcal{R} = \{(\vec X,\vec Y) \mid \vec X \sim_\prods \vec Y\}$. We want to show that $\mathcal{R}\subseteq{\simS_\prods}$ and $\mathcal{R}\subseteq{\lesssim_\prods^{-1}}$. 
	
	To show that $\mathcal{R}\subseteq{\simS_\prods}$, we show that for all $(\vec X,\vec Y)\in\mathcal{R}$: for every label $a\in\simSX\cup\simSZ$ and word $\vec X$, if $\transitionGNF{\prods}{\vec X}{a}{\vec X'}$, then there exists a word $\vec Y'$ s.t. $\transitionGNF{\prods}{\vec Y}{a}{\vec Y'}$ and $(\vec X',\vec Y')\in\mathcal{R}$ if $a\in\simSX$, or otherwise $(\vec Y',\vec X')\in\mathcal{R}$ if $a\in\simSZ$; and, for every label $a\in\simSY\cup\simSW$ and word $\vec Y$, if $\transitionGNF{\prods}{\vec Y}{a}{\vec Y'}$, then there exists a word $\vec X'$ s.t. $\transitionGNF{\prods}{\vec X}{a}{\vec X'}$ and $(\vec X',\vec Y')\in\mathcal{R}$ if $a\in\simSY$, or otherwise $(\vec Y',\vec X')\in\mathcal{R}$ if $a\in\simSW$. 
	
	First, let there be $a \in \simSX$ and $\vec X'$ s.t. $\transition{\vec X}{a}{\vec X'}$. Since $\vec X \sim_\prods \vec Y$, we know there is a $\vec Y'$ s.t. $\transitionGNF{\prods}{\vec Y}{a}{\vec Y'}$ and $\vec X' \sim_\prods \vec Y'$ and hence that $(\vec X',\vec Y')\in\mathcal{R}$. 
	
	Next, let there be $a\in\simSY$ and $\vec Y'$ s.t. $\transition{\vec Y}{a}{\vec Y'}$. Since $\vec X \sim_\prods \vec Y$, we know there is a $\vec X'$ s.t. $\transition{\vec X}{a}{\vec X'}$ and $\vec X'\sim_\prods \vec Y'$. Hence $(\vec X',\vec Y')\in\mathcal{R}$. 
	
	Now let there be $a\in\simSZ$ and $\vec X'$ s.t. $\transition{\vec X}{a}{\vec X'}$. Since $\vec X \sim_\prods \vec Y$, we know there is a $\vec Y'$ s.t. $\transition{\vec Y}{a}{\vec Y'}$ and $\vec X' \sim_\prods \vec Y'$. By the symmetry of $\sim_\prods$ we get $\vec Y' \sim_\prods \vec X'$, and therefore $(\vec Y',\vec X')\in\mathcal{R}$. 
	
	Finally, let there be $a\in\simSW$ and $\vec Y'$ s.t. $\transition{\vec Y}{a}{\vec Y'}$. Since $\vec X \sim_\prods \vec Y$, we know there is a $\vec X'$ s.t. $\transition{\vec X}{a}{\vec X'}$ and $\vec X'\sim_\prods \vec Y'$. By the symmetry of $\sim_\prods$ we get $\vec Y' \sim_\prods \vec X'$, and therefore $(\vec Y',\vec X')\in\mathcal{R}$.
	
	We proceed similarly to show that $\mathcal{R}\subseteq{\simS^{-1}_\prods}$. For all $(\vec X,\vec Y)\in\mathcal{R}$: for every label $a\in\simSX\cup\simSZ$ and type $\vec Y'$, if $\transitionGNF{\prods}{\vec Y}{a}{\vec Y'}$, then there exists a word $\vec X'$ s.t. $\transitionGNF{\prods}{\vec X}{a}{\vec X'}$ and $(\vec X',\vec Y')\in\mathcal{R}$ if $a\in\simSX$, or otherwise $(\vec Y',\vec X')\in\mathcal{R}$ if $a\in\simSZ$; and, for every label $a\in\simSY\cup\simSW$ and word $\vec X'$, if $\transitionGNF{\prods}{\vec X}{a}{\vec X'}$, then there exists a word $\vec Y'$ s.t. $\transitionGNF{\prods}{\vec Y}{a}{\vec Y'}$ and $(\vec X',\vec Y')\in\mathcal{R}$ if $a\in\simSY$, or otherwise $(\vec Y',\vec X')\in\mathcal{R}$ if $a\in\simSW$.
	
	First, let there be $a\in\simSX$ and $\vec Y'$ s.t. $\transition{\vec Y}{a}{\vec Y'}$. Since $\vec X\sim_\prods \vec Y$, we know there is a $\vec X'$ s.t. $\transitionGNF{\prods}{\vec X}{a}{\vec X'}$ and $\vec X'\sim_\prods \vec Y'$. Hence $(\vec X',\vec Y')\in\mathcal{R}$. 
	
	Next, let there be $a\in\simSY$ and $\vec X'$ s.t. $\transition{\vec X}{a}{\vec X'}$. Since $\vec X \sim_\prods U$, we know there is a $\vec Y'$ s.t. $\transitionGNF{\prods}{\vec Y}{a}{\vec Y'}$ and $\vec X'\sim_\prods \vec Y'$. Hence $(\vec X',\vec Y')\in\mathcal{R}$.
	
	Now let there be $a\in\simSZ$ and $\vec Y'$ s.t. $\transition{\vec Y}{a}{\vec Y'}$. Since $\vec X\sim \vec Y$, we know there is a $\vec X'$ s.t. $\transition{\vec X}{a}{\vec X'}$ and $\vec X'\sim_\prods \vec Y'$. By the symmetry of $\sim_\prods$ we get $\vec Y'\sim_\prods \vec X'$, and therefore $(\vec Y',\vec X')\in\mathcal{R}$.  
	
	Finally, let there be $a\in\simSW$ and $\vec X'$ s.t. $\transition{\vec X}{a}{\vec X'}$. Since $\vec X\sim \vec Y$ we know there is $\vec Y'$ s.t. $\transition{\vec Y}{a}{\vec Y'}$ and $\vec X'\sim \vec Y'$. By the symmetry of $\sim_\prods$ we get $\vec Y'\sim_\prods \vec X'$, and therefore $(\vec Y',\vec X')\in\mathcal{R}$.
\end{proof}
\begin{lemma}
	\label{lemma:bot-bisim}
	Let $T$ be a type and let $\mathcal{N}$ be the set of non-terminal symbols in the simple GNF grammar $\grm(T,\prods_0)$ for any $\prods_0$. For every $\vec X,\vec Y \in \mathcal{N}^*$ we have that  $\vec X \sim_\prods \vec X\bot\vec Y$.
\end{lemma}
\begin{proof}
	Immediate from the fact that $\bot$ has no productions.
\end{proof}
\begin{lemma}[The behaviours of types and their words coincide]
	\label{lemma:behaviour-types-words-coincide}
	Let $\wellFormed{}{T}$ and $(\vec{X}_T, \prods)=\grm(T,\prods_0)$ for some $\prods_0$. Then,
	\begin{itemize}
		\item If $\transition{T}{a}{U}$ for some $U$, then there exists $\vec{X}'$ such that $\transitionGNF{\prods}{\vec{X}_T}{a}{\vec{X}'}$ and $\vec X_U \sim_\prods \vec{X}'$, where $(\vec X_U,\prods') = \grm(U,\prods)$;
		\item If $\transitionGNF{\prods}{\vec{X}_T}{a}{\vec{X}'}$ for some $\vec{X}'$, then there exists $U$ such that $\transition{T}{a}{U}$ and $\vec X' \sim_\prods \vec X_U$, where $(\vec X_U,\prods') = \grm(U,\prods)$.
	\end{itemize}
\end{lemma}
\begin{proof}
	Let us define the relation
	\begin{gather*}
		\mathcal{R} = \{(T,\vec{Y}) \mathrel{|} \text{$\vec{X}_T \sim_\prods \vec{Y}$, where $(\vec{X}_T,\prods)=\grm(T,\emptyset)$}\}.
	\end{gather*}
	We now prove that $\mathcal{R}$ is backward closed for the transition relations, showing the desired property. We begin by considering the cases in which $T$ fits a type constructor, i.e., $T = \unravel(T)$.
	We have the following case analysis for $T$.
	
	(\textbf{Case} $T=\tUnit$): by rule \textsc{L-Unit}, the LTS at $T$ has the unique transition $\transition{T}{\tUnit}{\tSkip}$. Similarly, the LTS at $\vec X_T$ has the unique transition $\transitionGNF{\prods}{\vec X_T}{\tUnit}{\varepsilon}$. Since $\vec X_T \sim_\prods \vec Y$, the LTS at $\vec Y$ has the unique transition $\transitionGNF{\mathcal P}{\vec Y}{\tUnit}{\vec Y'}$ for some $\vec Y$ such that $\vec Y \sim_\prods \varepsilon$. Since $\word(\tSkip)=\varepsilon$, we have $\word(\tSkip) \sim_\prods \vec Y'$ and therefore $(\tSkip, \vec Y')\in \mathcal{R}$.
	
	(\textbf{Case} $T=\tArrow{U}{*}{V}$): By rules \textsc{L-ArrowDom} and \textsc{L-ArrowRng}, we know that $T$ has exactly two transitions: $\transition{T}{\lArrowDomain}{U}$ and $\transition{T}{\lArrowRange}{V}$ (rule \textsc{L-LinArrow} does not apply). Similarly, the LTS at $\vec X_T$ has exactly two transitions: $\transitionGNF{\prods}{\vec X_T}{\lArrowDomain}{\word(U)}$ and $\transitionGNF{\prods}{\vec X_T}{\lArrowRange}{\word(V)}$. Since $\vec{X}_T\sim_\prods \vec Y$, the LTS at $\vec Y$ has exactly two transitions, $\transitionGNF{\prods}{\vec Y}{\lArrowDomain}{\vec Y_1}$ and $\transitionGNF{\prods}{\vec Y}{\lArrowDomain}{\vec Y_2}$ for some $\vec Y_1,\vec Y_2$ s.t. $\vec Y_1\sim_\prods \word(U)$ and $\vec Y_2\sim_\prods\word(V)$. Hence $(U,\vec Y_1),(V, \vec Y_2)\in\mathcal{R}$.
	
	(\textbf{Case} $T=\tArrow{U}{{\lin}}{V}$): Similar to the previous case, but the LTS at $T$ has a single additional transition in $\transition{T}{\lArrowLinear}{\tSkip}$, as does the LTS at $\vec X_T$ in $\transitionGNF{\prods}{\vec X_T}{\lArrowLinear}{\varepsilon}$. Since $\vec X_T \sim_\prods \vec Y$, we know $\vec Y$ also has an additional transition, $\transitionGNF{\prods}{\vec Y}{\lArrowLinear}{\vec Y'}$ for some $\vec Y'$ s.t. $\vec Y'\sim_\prods\varepsilon$. Since $\word(\tSkip)=\varepsilon$, we have $(\tSkip,\vec Y')\in\mathcal{R}$.
	
	(\textbf{Case} $T = \tRcdVrt{\ell}{T}{L}$): By rules \textsc{L-RcdVrt} and \textsc{L-RcdVrtField}, the LTS at $T$ has exactly the transitions $\transition{T}{\llparenthesis\rrparenthesis_\checkmark}{\tSkip}$ and $\transition{T}{\llparenthesis\rrparenthesis_k}{T_k}$ for each $k\in L$. Similarly, the LTS at $\vec X_T$ has exactly the transitions  $\transitionGNF{\prods}{\vec X_T}{\llparenthesis\rrparenthesis_\checkmark}{\bot}$ and $\transitionGNF{\prods}{\vec X_T}{\llparenthesis\rrparenthesis_k}{\word(T_k)}$ for each $k\in L$. Since $\vec X_T \sim_\prods \vec Y$, the LTS at $\vec Y$ has exactly the transitions $\transitionGNF{\prods}{\vec Y}{\llparenthesis\rrparenthesis_k}{\vec Y_k}$ for some $\vec Y_k$ s.t. $\word(T_k)\sim_\prods \vec Y$ for each $k \in L$ and additionally $\transitionGNF{\prods}{\vec Y}{\llparenthesis\rrparenthesis_\checkmark}{\vec Y_\checkmark}$ for some $\vec Y_\checkmark$ s.t. $\vec Y_\checkmark \sim_\prods \bot$. It follows that $(T_k, \vec Y_k)\in\mathcal{R}$ for each $k\in L$. Observe that, since neither $\varepsilon$ or $\bot$ have any transitions, we have $\varepsilon \sim_\prods \bot$. From this and $\word(\tSkip)=\varepsilon$ it follows that $(\tSkip, \vec Y_\checkmark)\in \mathcal{R}$.
	
	(\textbf{Case} $T=x$): Cannot occur, since $\not\wellFormed{}{x}$.
	
	(\textbf{Case} $T=\tEnd$): by rule \textsc{L-End}, the LTS at $T$ has the unique transition $\transition{T}{\tEnd}{\tSkip}$. Similarly, the LTS at $\vec X_T$ has the unique transition $\transitionGNF{\prods}{\vec X_T}{\tEnd}{\bot}$. Since $\vec X_T \sim_\prods \vec Y$, the LTS at $\vec Y$ has the unique transition $\transitionGNF{\mathcal P}{\vec Y}{\tEnd}{\vec Y'}$ for some $\vec Y$ such that $\vec Y \sim_\prods \bot$. Since $\word(\tSkip)=\varepsilon$ and $\varepsilon \sim_\prods \bot$, we have $\word(\tSkip) \sim_\prods \vec Y'$ and therefore $(\tSkip, \vec Y')\in \mathcal{R}$.
	
	(\textbf{Case}  $T=\tMsg{U}$): By rule \textsc{L-Msg}, the LTS at $T$ has exactly two transitions $\transition{T}{\lMsgPayload}{U}$ and $\transition{T}{\lMsgContinuation}{\tSkip}$. Similarly, the LTS at $\vec X_T$ has exactly two transitions $\transitionGNF{\prods}{\vec X_T}{\lMsgPayload}{\word(U)\bot}$ and $\transitionGNF{\prods}{\vec X_T}{\lMsgContinuation}{\varepsilon}$. Since $\vec X_T \sim_\prods \vec Y$, the LTS at $\vec Y$ has exactly two transitions $\transitionGNF{\prods}{\vec Y}{\lMsgPayload}{\vec Y_1}$ and $\transitionGNF{\prods}{\vec Y}{\lMsgContinuation}{\vec Y_2}$ for some $\vec Y_1,\vec Y_2$ s.t. $\word(U)\bot \sim_\prods \vec Y_1$ and $\varepsilon \sim_\prods \vec Y_2$. By Lemma \ref{lemma:bot-bisim}, we have $\word(U) \sim_\prods \word(U)\bot \sim_\prods \vec Y_1$, and by the definition of $\grm$, $\word(\tSkip) \sim_\prods \varepsilon \sim_\prods \vec Y_2$. Hence $(U, \vec Y_1), (\tSkip, \vec Y_2) \in \mathcal{R}$. 
	
	(\textbf{Case} $T = \tChoice{\ell}{T}{L}$): By rules \textsc{L-RcdVrt} and \textsc{L-RcdVrtField}, the LTS at $T$ has exactly the transitions $\transition{T}{\choicebranch_\checkmark}{\tSkip}$ and $\transition{T}{\choicebranch_k}{T_k}$ for each $k\in L$. Similarly, the LTS at $\vec X_T$ has transitions $\transitionGNF{\prods}{\vec X_T}{\choicebranch_\checkmark}{\bot}$ and $\transitionGNF{\prods}{\vec X_T}{\choicebranch_k}{\word(T_k)}$ for each $k\in L$. Since $\vec X_T \sim_\prods \vec Y$, the LTS at $\vec Y$ has exactly the transitions $\transitionGNF{\prods}{\vec Y}{\choicebranch_\checkmark}{\vec Y_\checkmark}$ for some $\vec Y_\checkmark$ s.t. $\vec Y_\checkmark \sim_\prods \bot$ and $\transitionGNF{\prods}{\vec Y}{\choicebranch_k}{\vec Y_k}$ for some $\vec Y_k$ s.t. $\word(T_k)\sim_\prods \vec Y_k$ for each $k \in L$. Since $\bot$ has no productions, we have $\bot \sim_\prods \varepsilon$ and by transitivity $\vec Y \sim_\prods \varepsilon$. Hence $(\tSkip, \vec Y_\checkmark)\in \mathcal{R}$ and $(T_k, \vec Y_k) \in \mathcal{R}$ for each $k\in L$.
	
	(\textbf{Case} $T=\tSeq{\inout U}{S}$): By rules \textsc{L-MsgSeq1} and \textsc{L-MsgSeq2}, the LTS at $T$  has exactly two transitions, $\transition{T}{\lMsgPayload}{U}$ and $\transition{T}{\lMsgContinuation}{S}$. Given that $\word(T) = \word(\inout U)\cdot\word(S)$ and that $\word(\inout U)$ yields exactly two productions $\word(\inout U) \to \lMsgPayload \word(U) \bot$ and $\word(\inout U) \to \lMsgContinuation \word(S)$, we have for the LTS at $\vec X_T$ exactly two transitions $\transitionGNF{\mathcal{P}}{X_T}{\lMsgPayload}{\word(U) \bot \word(S)}$ and $\transitionGNF{\mathcal{P}}{X_T}{\lMsgContinuation}{\word(S)}$. Since $X_T \sim_\prods \vec Y$, the LTS at $\vec Y$ has exactly two transitions $\transitionGNF{\prods}{\vec Y}{\lMsgPayload}{\vec Y_1}$ and $\transitionGNF{\prods}{\vec Y}{\lMsgContinuation}{\vec Y_2}$ for some $\vec Y_1, \vec Y_2$ s.t. $\word(U) \bot \word(S) \sim_\prods \vec Y_1$ and $\word(S) \sim_\prods \vec Y_2$. By Lemma \ref{lemma:bot-bisim}, we have $\word(U) \sim_\prods \word(U)\bot \word(S) \sim_\prods \vec Y_1$. Hence $(U,\vec Y_1),(S, \vec Y_2)\in\mathcal{R}$.
	
	(\textbf{Case} $T=\tSeq{x}{S}$): Cannot occur, since $\not\wellFormed{}{\tSeq{x}{S}}$.
	
	Now we consider the cases in which $T \neq \unravel(T)$. It is straightforward to check that the LTS at $T$ has a transition $\transition{T}{a}{U}$ iff the LTS at $\unravel(T)$ has a corresponding transition $\transition{\unravel(T)}{a}{U}$ (with the same $U$). This is a consequence of the fact that the LTS rules for $T$ essentially follow its unfolding, which eventually terminates due to contractivity. On the side of grammars, we have $\word(\unravel(T)) \sim_\prods \word(T) \sim_\prods \vec Y$. Now suppose that $\transition{T}{a}{U}$; then $\transition{\unravel(T)}{a}{U}$; since $\word(\unravel(T)) \sim_\prods \vec Y$, our previous case analysis yields $\transition{\vec Y}{a}{\vec Y'}$ for some $\vec Y'$ s.t. $\word(U)\sim_\prods\vec Y'$; concluding that $(U, \vec Y') \in \mathcal{R}$. Conversely, suppose that $\transitionGNF{\prods}{\vec Y}{a}{\vec Y'}$ for some $\vec Y'$. Since $\word(\unravel(T)) \sim_\prods \vec Y$, our previous case analysis yields $\transitionGNF{\prods}{\unravel(T)}{a}{U}$ for some $U$ s.t. $\word(U) \sim_\prods \vec Y'$. We conclude that $\transition{T}{a}{U}$ and therefore that $(U, \vec Y') \in \mathcal{R}$.
\end{proof}
We are now able to prove \cref{thm:soundness-grammars}.
\sfg*
\begin{proof}
	Consider the following relation on pairs of types
	\begin{gather*}
		\mathcal{R}=\{(V,W) \,|\, \vec{X}_V \simS_\prods \vec{X}_W\}\text{,}
	\end{gather*}
	where $\wellFormed{}{V}$, $\wellFormed{}{W}$, $(\vec{X}_V,\prods') = \grm(V,\emptyset)$ and $(\vec{X}_W,\prods)=\grm(W,\prods')$.
	
	To prove the desired property, we show that $\mathcal{R}\subseteq{\simS}$, in other words, that $\mathcal{R}$ is an $\mathcal{XYZW}$-simulation with $\mathcal{X,Y,Z,W}$ defined as the sets generated, respectively, by the grammars for $a_\mathcal{X},a_\mathcal{Y},a_\mathcal{Z},a_\mathcal{W}$ in Definition \ref{def:hocfst-sim}. Take $(V,W)\in\mathcal{R}$ with $(\vec{X}_V,\prods') = \grm(V,\emptyset)$ and $(\vec{X}_W,\prods)=\grm(W,\prods')$.
	
	First, let there be $a\in \simSX$ and $V'$ s.t. $\transition{V}{a}{V'}$. We want to show that there is a $W'$ s.t. $\transition{W}{a}{W'}$ and $(V',W')\in\mathcal{R}$. By Lemma \ref{lemma:behaviour-types-words-coincide}, there exists $\vec{Y}$ such that $\transitionGNF{\prods}{\vec{X}_V}{a}{\vec{Y}}$ and $\word(V')\sim_\prods\vec{Y}$. Since $\vec{X}_V \simS_\prods \vec{X}_W$, we know there exists a matching word $\vec{Z}$ such that $\transitionGNF{\prods}{\vec{X}_W}{a}{\vec{Z}}$ and that $\vec{Y} \simS_\prods \vec{Z}$. Again by Lemma \ref{lemma:behaviour-types-words-coincide}, there exists $W'$ such that $\transition{W}{a}{W'}$ and $\vec Z \sim_\prods \word(W')$. Let $(\vec{X}_{V'}, \prods') = \grm(V', \emptyset)$ and $(\vec{X}_{W'}, \prods) = \grm(W', \prods')$, and recall that $\vec{X}_{V'} = \word(V')$ and $\vec{X}_{W'} = \word(W')$. Having $\vec{X}_{V'} \sim_\prods \vec{Y}$ and $\vec{X}_{W'} \sim_\prods \vec{Z}$, we know by Lemma \ref{lemma:bisim-implies-sim} that $\vec{X}_{V'} \simS_\prods \vec{Y}$ and $\vec{Z} \simS_\prods \vec{X}_{W'}$. Recalling that $\vec Y \simS_\prods \vec Z$, we establish by transitivity that $\vec{X}_{V'}\simS_{\prods}\vec{X}_{W'}$ and therefore that $(V',W')\in\mathcal{R}$. 
	
	Next, let there be $a \in \simSY$ and $W'$ s.t. $\transition{W}{a}{W'}$. We want to show that there is a $V'$ s.t. $\transition{V}{a}{V'}$ and $(V',W')\in\mathcal{R}$. Lemma \ref{lemma:behaviour-types-words-coincide} gives us some $\vec{Z}$ with $\transitionGNF{\prods}{\vec{X}_W}{a}{\vec{Z}}$ and $\word(W')\sim_\prods\vec{Z}$. Since $a\in\simSY$, we know from $\vec{X}_V \simS_\prods \vec{X}_W$ that there is a matching word $\vec{Y}$ such that $\transitionGNF{\prods}{\vec{X}_V}{a}{\vec{Y}}$ and that $\vec{Y}\simS_\prods \vec{Z}$. Once again by Lemma \ref{lemma:behaviour-types-words-coincide} we know there exists $V'$ such that $\transition{V}{a}{V'}$ and $\vec Y \sim_\prods \word(V')$. Let $(\vec{X}_{V'}, \prods') = \grm(V', \emptyset)$ and $(\vec{X}_{W'}, \prods) = \grm(W', \prods')$. Having $\vec Y \sim_\prods \vec{X}_{V'}$ and $\vec{X}_{W'} \sim_\prods \vec{Z}$, we know by Lemma \ref{lemma:bisim-implies-sim} that $\vec{X}_{V'} \simS_\prods \vec Y$ and $\vec{Z} \simS_\prods \vec{X}_{W'}$. Recalling that $\vec{Y}\simS_\prods\vec{Z}$, we can establish by transitivity that $\vec{X}_{V'}\simS_{\prods}\vec{X}_{W'}$ and therefore that $(V',W')\in\mathcal{R}$.
	
	Now let there be $a \in \simSZ$ and $V'$ s.t. $\transition{V}{a}{V'}$. We want to show that there is a $W'$ s.t. $\transition{W}{a}{W'}$ and $(W',V')\in\mathcal{R}$. By Lemma \ref{lemma:behaviour-types-words-coincide}, there exists $\vec Y$ s.t. $\transition{\vec X_V}{a}{\vec Y}$ and $\word(V')\sim_\prods \vec Y$. Since $\vec X_V \simS_\prods \vec X_W$, we know there exists a matching word $\vec Z$ s.t. $\transition{\vec X_W}{a}{\vec Z}$ and that $\vec Z \simS_\prods \vec Y$. Again by Lemma \ref{lemma:behaviour-types-words-coincide}, there exists $W'$ s.t. $\transition{W}{a}{W'}$ and $\vec Z \sim_\prods \word(W')$. Let $(\vec{X}_{V'}, \prods') = \grm(V', \emptyset)$ and $(\vec{X}_{W'}, \prods) = \grm(W', \prods')$, and recall that $\vec{X}_{V'} = \word(V')$ and $\vec{X}_{W'} = \word(W')$. Having $\vec{X}_{W'} \sim_\prods \vec{Z}$ and $\vec{X}_{V'} \sim_\prods \vec{Y}$, we know by Lemma \ref{lemma:bisim-implies-sim} that $\vec{X}_{W'} \simS_\prods \vec{Z}$ and $\vec{Y} \simS_\prods \vec{X}_{V'}$. Recalling that $\vec Z \simS_\prods \vec Y$, we establish by transitivity that $\vec{X}_{W'}\simS_{\prods}\vec{X}_{V'}$ and therefore that $(W',V')\in\mathcal{R}$. 
	
	Finally, let there be $a \in \simSW$ and $W'$ s.t. $\transition{W}{a}{W'}$. We want to show that there is a $V'$ s.t. $\transition{V}{a}{V'}$ and $(V',W')\in\mathcal{R}$. Lemma \ref{lemma:behaviour-types-words-coincide} gives us some $\vec{Z}$ with $\transitionGNF{\prods}{\vec{X}_W}{a}{\vec{Z}}$ and $\word(W')\sim_\prods\vec{Z}$. Since $a\in\simSY$, we know from $\vec{X}_V \simS_\prods \vec{X}_W$ that there is a matching word $\vec{Y}$ such that $\transitionGNF{\prods}{\vec{X}_V}{a}{\vec{Y}}$ and that $\vec{Z}\simS_\prods \vec{Y}$. Once again by Lemma \ref{lemma:behaviour-types-words-coincide} we know there exists $V'$ such that $\transition{V}{a}{V'}$ and $\vec Y \sim_\prods \word(V')$. Let $(\vec{X}_{V'}, \prods') = \grm(V', \emptyset)$ and $(\vec{X}_{W'}, \prods) = \grm(W', \prods')$. Having $\vec{X}_{W'} \sim_\prods \vec{Z}$ and $\vec Y \sim_\prods \vec{X}_{V'}$, we know by Lemma \ref{lemma:bisim-implies-sim} that $\vec{X}_{W'} \simS_\prods \vec{Z}$ and $\vec{Y} \simS_\prods \vec{X}_{V'}$. Recalling that $\vec{Z}\simS_\prods\vec{Y}$, we can establish by transitivity that $\vec{X}_{V'}\simS_{\prods}\vec{X}_{W'}$ and therefore that $(V',W')\in\mathcal{R}$.
\end{proof}
\section{Pruning}
The grammars generated by procedure $\grm$ may contain unreachable words, which can be ignored by the algorithm. Intuitively, these words correspond to communication actions that cannot be fulfilled, such as part $\tIn{\tBool}$ in type $\tSeq{(\tRec{s}{\tSeq{\tOut{\tInt}}{s}})}{\tIn{\tBool}}$. Formally, these words appear in productions following what are known as \textit{unnormed words}.
\begin{definition} 
	Let $\vec a$ be a non-empty sequence of non-terminal symbols $a_1,\ldots,a_n$. Write $\transitionGNF{\prods}{\vec Y}{\vec a}{\vec Z}$ when $\transitionGNF{\prods}{\vec Y}{a_1}{\transitionGNF{\prods}{\ldots}{a_n}{\vec Z}}$. We say that a word $\vec Y$ is \textit{normed} if $\transitionGNF{\prods}{\vec Y}{\vec a}{\varepsilon}$ for some $\vec a$, and \textit{unnormed} otherwise. If $\vec Y$ is normed and $\vec a$ is the shortest path such that $\transitionGNF{\prods}{\vec Y}{\vec a}{\varepsilon}$, then $\vec a$ is called the \textit{minimal path} of $\vec Y$, and its length is the \textit{norm} of $\vec Y$, denoted $\fun{norm}(\vec Y)$.
\end{definition}
It is known that any unnormed word $\vec Y$ is bisimilar to its concatenation with any other word, i.e., if $\vec Y$ is unnormed, then $\vec Y \sim_\prods \vec Y \vec X$. It is easy to show that ${\sim_\prods} \subseteq {\simS_\prods}$, and hence that $\vec Y \simS_\prods \vec Y \vec X$. In this case, $\vec X$ is said to be unreachable and can be safely removed from the grammar. We call the procedure of removing all unreachable symbols from a grammar \textit{pruning}, and denote the pruned version of a grammar $\prods$ by $\prune(\prods)$.
\pruningPreservesXYZWSimilarity*
\begin{proof}
	For the direct implication, the $\mathcal{XYZW}$-simulation for $\vec X$ and $\vec Y$ over $\prods$ is also an $\mathcal{XYZW}$-simulation for $\vec X$ and $\vec Y$ over $\prune(\prods)$. For the reverse implication, if $\mathcal{R}'$ is an $\mathcal{XYZW}$-simulation for $\vec X$ and $\vec Y$ over $\prune(P)$, then relation
	\begin{gather*}
		\mathcal{R} = \mathcal{R}' \cup \{(\vec V W, \vec V W \vec Z) \,|\, (W \to
		\vec V W \vec Z) \in \prods, W \,\text{unnormed}\}
	\end{gather*}
	is an $\mathcal{XYZW}$-simulation for $\vec X$ and $\vec Y$ over $\mathcal{P}$.
\end{proof}
\section{Proof of Theorem \ref{thm:soundness-alg}}
To prove the soundness of our algorithm, we need to prove the soundness of its three phases: translation from types to grammars, grammar pruning and exploration of an $\mathcal{XYZW}$-expansion tree (with $\mathcal{X}$, $\mathcal{Y}$, $\mathcal{Z}$ and $\mathcal{W}$ instantiated as in \cref{def:hocfst-sim}). The soundness of the first two phases is given by \cref{thm:soundness-grammars} and \cref{thm:pruning-preserves-xyzw}, respectively. It remains now to prove that the last phase is also sound.

The exploration of the $\mathcal{XYZW}$-expansion tree is carried out through expansion and simplification steps. Expansion steps attempt to build an $\mathcal{XYZW}$-simulation from an initial pair of grammar words $(\vec X, \vec Y)$, while simplification steps are used modify the construction of the tree, attempting to keep some branches of the tree finite even when the corresponding $\mathcal{XYZW}$-simulation is not. This procedure returns $\True$ whenever an empty node is reached (meaning there is a successful branch), or $\False$ if all nodes fail to expand (meaning there is no successful branch). Given these stopping conditions, the soundness of the exploration procedure relies on the tree it yields having the following property: there is a successful branch iff $\vec{X} \xyzwsimP \vec{Y}$. This is known as the \emph{safeness property}.

The proof of this property is given in two parts: in the first we show the property holds for a tree constructed through expansion steps only; in the second we show that the simplification rules modify the tree safely, i.e., if their application results in a tree with a successful branch, then an $\mathcal{XYZW}$-simulation can actually be constructed.
\sp*
\begin{proof}
	In an expansion tree without simplification both directions follow directly from the definition of expansion. Recall that in an $\mathcal{XYZW}$-expansion tree, each child is an $\mathcal{XYZW}$-expansion of its parent. Observe then that in an $\mathcal{XYZW}$-expansion tree rooted at $\{(\vec X, \vec Y)\}$, the union of all nodes along a successful branch (i.e. infinite or containing an empty leaf) constitutes an $\mathcal{XYZW}$-simulation (over productions $\prods$) that includes $(\vec X, \vec Y)$. Hence $\vec X \simS_\prods Y$.
	
	In an expansion tree with simplifications, the union of all nodes along a successful branch need not be an $\mathcal{XYZW}$-simulation, only a (hopefully finite) representation of it, i.e., a set with which we can reconstruct it by reversing the simplifications. It remains to show how this can be done for each simplification rule:
	\begin{itemize}
		\item \textsc{Reflexivity}: Let $N$ be a node at depth $n$ such that $\{(\vec X_i, \vec X_i)\}_{i\in1..j}\subseteq N$ for some $j$. Applying \textsc{Reflexivity}, its simplification is $N' = N\setminus\{(\vec X_i, \vec X_i)\}_{i\in1..j}$. Observe that the reflexive closure of the union of all nodes along the successful branch containing $N'$ is an $\mathcal{XYZW}$-simulation containing $N$.
		
		\item \textsc{Preorder}: Let $\leq_N$ be the least preorder containing the ancestors of a node $N$. Applying \textsc{Preorder}, its simplification is $N'=N\setminus{\leq_N}$. Observe that the reflexive and transitive closure of the union of all nodes along the successful branch containing $N'$ is an $\mathcal{XYZW}$-simulation containing $N$.
		
		\item \textsc{Split}: Let $N$ be a node containing a pair of the form $(X_0 \vec X, Y_0 \vec Y)$ with $\fun{norm}(X_0) \leq \fun{norm}(Y_0)$ (the case where $\fun{norm}(X_0) > \fun{norm}(Y_0)$ is similar). Let sequence $\vec a = a_1,\dots,a_k$ be a minimal path for $X_0$, and $\vec Z$ be a word such that $\transitionGNF{\prods}{Y_0}{\vec a}{\vec Z}$. Applying \textsc{Split} to $N$ yields $N$ itself and a sibling $N'$ with pairs $(X_0\vec Z, Y_0),(\vec X, \vec Z\vec Y)$ in place of $(X_0 \vec X, Y_0 \vec Y)$. We need to show that, assuming there is an $\mathcal{XYZW}$-simulation over $\prods$ containing $N'$, it is possible to obtain an $\mathcal{XYZW}$-simulation over $\prods$ containing $N$.
		
		Let $\mathcal{R'}$ be an $\mathcal{XYZW}$-simulation over $\prods$ containing $N'$ and $\mathcal{S}_1$ be the smallest $\mathcal{XYZW}$-simulation over $\prods$ that includes pair $(X_0\vec Z, Y_0)$. Then, relation 
		$$\mathcal{R} = \mathcal{R'} \cup \{(X_0 \vec X, Y_0 \vec Y)\} \cup \{(\vec X_1 \vec X, \vec Y_1 \vec Y) \mathrel{|} (\vec X_1 \vec Z,\vec Y_1) \in \mathcal{S}_1\}$$ 
		is an $\mathcal{XYZW}$-simulation over $\prods$ containing $N$. 
		
		Since $N'$ includes all pairs of $N$ except $(X_0\vec X, Y_0\vec Y)$ and $\mathcal{R}'$ is an $\mathcal{XYZW}$-simulation over $\prods$ containing $N'$, it follows from the union of $\mathcal{R}'$ with $\{(X_0 \vec X, Y_0 \vec Y)\}$ that $\mathcal{R}$ contains $N$. 
		
		All that remains now is to show that $\mathcal{R}$ is an $\mathcal{XYZW}$-simulation over $\prods$. Since $\mathcal{R}'$ is already such a relation, we need to show that it remains so after adding to it every pair in $\{(X_0\vec X, Y_0\vec Y)\}\cup\{(\vec X_1 \vec X, \vec Y_1 \vec Y) \mathrel{|} (\vec X_1 \vec Z, \vec Y_1) \in \mathcal{S}_1\}$. For $(X_0\vec X, Y_0\vec Y)$ this is easy: we observe that the pairs containing the derivatives of all matching transitions of $X_0\vec X$ and $Y_0\vec Y$ are elements of $\{(\vec X_1 \vec X, \vec Y_1 \vec Y) \mathrel{|} (\vec X_1 \vec Z, \vec Y_1) \in \mathcal{S}_1\}$. For the pairs in $\{(\vec X_1 \vec X, \vec Y_1 \vec Y) \mathrel{|} (\vec X_1 \vec Z, \vec Y_1) \in \mathcal{S}_1\}$ we need to distinguish two cases:
		\begin{itemize}
			\item (\textbf{Case} $\vec X_1 \neq \varepsilon$): We observe that, from the definition of $\mathcal{S}_1$, it follows that relation $\{(\vec X_1 \vec X, \vec Y_1 \vec Y) \mathrel{|} (\vec X_1 \vec Z, \vec Y_1) \in \mathcal{S}_1\}$, which is contained in $\mathcal{R}$, contains the paired derivatives of all matching transitions of $X_1\vec X$ and $Y_1 \vec Y$.
			\item (\textbf{Case} $\vec X_1 = \varepsilon$): We observe that, since $\mathcal{R}$ contains $\mathcal{R}'$ (an $\mathcal{XYZW}$-simulation containing $N'$), it must include the paired derivatives of all matching transitions of $\vec X$ and $Y_1 \vec Y$.  
		\end{itemize}
	\end{itemize}
\end{proof}
With the safeness property established, we can finally put together a soundness proof encompassing all of the three phases of the algorithm.
\begin{lemma}
	\label{thm:soundness-pruned}
	If $\subG(\vec X_T, \vec X_U,
	\prune(\prods))$ returns \True, then $\vec X_T \simS_{\prune(\prods)} \vec X_U$.
\end{lemma}
\begin{proof}
	Function $\subG$ returns $\True$ whenever it finds a finite successful branch (i.e., a branch terminating in
	an empty node) in the expansion tree rooted at $\{(\vec X, \vec Y )\}$. Conclude with the safeness property, \cref{lemma:safeness}.
\end{proof}
\soundness*
\begin{proof}
	From \cref{thm:soundness-grammars}, \cref{thm:pruning-preserves-xyzw} and \cref{thm:soundness-pruned}.
\end{proof}
\section{Generating subtyping pairs}
\label{sec:appendix-subtyping-pairs}
We rely on a number of properties of subtyping to generate valid and invalid subtyping pairs. Before enumerating these properties, we introduce the following definition.
\begin{definition}
	The sets of free references in covariant and contravariant positions in a type $T$, respectively $\mathsf{freeCov}(T)$ and $\mathsf{freeContrav}(T)$, are mutually defined by induction on the structure of $T$:
	\begin{align*}
		\mathsf{freeCov}(\tArrow{T}{m}{U}) &= \mathsf{freeContrav}(T) \cup \mathsf{freeCov}(U)\\ 
		\mathsf{freeCov}(\tRcdVrt{\ell}{T_\ell}{L}) &= \bigcup_{k\in L} \mathsf{freeCov}(T_k)\\
		\mathsf{freeCov}(\tIn{T}) &= \mathsf{freeCov}(T)\\
		\mathsf{freeCov}(\tOut{T}) &= \mathsf{freeContrav}(T)\\
		\mathsf{freeCov}(\tSeq{S}{R}) &= \mathsf{freeCov}(S)\cup\mathsf{freeCov}(R)\\
		\mathsf{freeCov}(\tRec{x}{T}) &= \mathsf{freeCov}(T) \setminus \{x\}\\
		\mathsf{freeCov}(x) &= \{x\}
	\end{align*}
	(and in all other cases by $\mathsf{freeCov}(T) = \emptyset$)
	\begin{align*} 
		\mathsf{freeContrav}(\tArrow{T}{m}{U}) &= \mathsf{freeCov}(T) \cup \mathsf{freeContrav}(U)\\ 
		\mathsf{freeContrav}(\tRcdVrt{\ell}{T_\ell}{L}) &= \bigcup_{k\in L} \mathsf{freeContrav}(T_k)\\
		\mathsf{freeContrav}(\tIn{T}) &= \mathsf{freeContrav}(T)\\
		\mathsf{freeContrav}(\tOut{T}) &= \mathsf{freeCov}(T)\\
		\mathsf{freeContrav}(\tSeq{S}{R}) &= \mathsf{freeContrav}(S)\cup\mathsf{freeContrav}(R)\\
		\mathsf{freeContrav}(\tRec{x}{T}) &= \mathsf{freeContrav}(T) \setminus \{x\}
	\end{align*}
	(and in all other cases by $\mathsf{freeContrav}(T) = \emptyset$)
\end{definition}
The following theorem enumerates the properties of the subtyping relation. From it we can derive generation algorithm for valid subtyping pairs, parameterized on the size $i$ of the pair: \textit{if $i=0$, then select one of the pairs in \cref{thm:props-base}; if $i\ge 1$, then select one of the pairs in the remaining items.}
\begin{theorem}[Properties of subtyping (valid)]
	\label{thm:props-subtyping-valid}\
	\begin{enumerate}
		\item\label{thm:props-base} $\tUnit \subS \tUnit$, $\tEnd \subS \tEnd$ and $\tSkip \subS \tSkip$;
		\item $\tArrow{T}{m}{U} \subS \tArrow{V}{n}{W}$ if $V\subS T$, $U\subS W$ and $m\sqsubseteq n$;
		\item $\tRecord{\ell}{T_\ell}{L} \subS \tRecord{k}{U_k}{L}$ if $K \subseteq L$ and $T_j\subS U_j (\forall j\in K)$;
		\item $\tVariant{\ell}{T_\ell}{L}\subS \tVariant{k}{U_k}{K}$ if $L \subseteq K$ and $T_j \subS U_j (\forall j\in L)$; 
		\item \label{item:valid-5} $\tSeq{S_1}{S_2} \subS \tSeq{R_1}{R_2}$ if $S_1 \subS R_1$ and $S_2 \subS R_2$; 
		\item $\tRec{x}{T} \subS \tRec{x}{U}$ if:
		\begin{enumerate}
			\item \label{thm:props-rec} $T\subS U$ and $x\notin \fun{freeContrav}(T)\cup\fun{freeContrav}(U)$; 
			\item $T \subS U$ and $U \subS T$; 
		\end{enumerate}
		\item $\tIntChoice{\ell}{S_\ell}{L} \subS \tIntChoice{k}{R_k}{L}$ if $K \subseteq L$ and $S_j\subS R_j (\forall j\in L)$;
		\item $\tExtChoice{\ell}{S_\ell}{L}\subS \tExtChoice{k}{R_k}{K}$ if $L \subseteq K$ and $S_j \subS R_j (\forall j\in L)$;
		\item $\tSeq{\tEnd}{S}\subS\tEnd$, $\tEnd\subS\tSeq{\tEnd}{S}$ and $\tSeq{\tEnd}{S}\subS\tSeq{\tEnd}{R}$ for any $S,R$;
		\item $\tSeq{S}{\tSkip}\subS R$, $\tSeq{\tSkip}{S}\subS R$, $S\subS\tSeq{R}{\tSkip}$ and $S\subS\tSeq{\tSkip}{R}$ if $S\subS R$;
		\item $\tSeq{\tIntChoice{\ell}{S_\ell}{L}}{S'} \subS \tIntChoice{k}{\tSeq{R_k}{R'}}{L}$ and $\tIntChoice{\ell}{\tSeq{S_\ell}{S'}}{L} \subS \tSeq{\tIntChoice{k}{R_k}{L}}{R'}$ if $K \subseteq L$, $S_j\subS R_j (\forall j\in L)$ and $S'\subS R'$;
		\item $\tSeq{\tExtChoice{\ell}{S_\ell}{L}}{S'} \subS \tExtChoice{k}{\tSeq{R_k}{R'}}{L}$ and $\tExtChoice{\ell}{\tSeq{S_\ell}{S'}}{L} \subS \tSeq{\tExtChoice{k}{R_k}{L}}{R'}$ if $L \subseteq K$, $S_j\subS R_j (\forall j\in L)$ and $S'\subS R'$;
		\item $\tSeq{S_1}{(\tSeq{S_2}{S_3})} \subS \tSeq{(\tSeq{R_1}{R_2})}{R_3}$ and $\tSeq{(\tSeq{S_1}{S_2})}{S_3} \subS \tSeq{R_1}{(\tSeq{R_2}{R_3})}$ if $S_1\subS R_1$, $S_2\subS R_2$ and $S_3 \subS R_3$;
		\item $\tRec{x}{T}\subS U$ if $T\subS U$ and $x \notin \free(T)$;
		\item $T\subS \tRec{x}{U}$ if $T\subS U$ and $x \notin \free(U)$;
		\item $\tRec{x}{T}\subS \subst{y}{\tRec{y}{U}}{U}$ if $\tRec{x}{T}\subS \tRec{y}{U}$.
	\end{enumerate}
\end{theorem}
\begin{proof}
	By observation of the syntactic subtyping rules and the definition of $\mathsf{freeContra}$ for \cref{thm:props-rec}.
\end{proof}

Identifying the set of references in contravariant positions makes it easier to generate valid subtyping pairs featuring recursion. Observe that, despite looking so, $\tRec{t}{\tArrow{t}{*}{t}}$ is not a subtype of $\tRec{t}{\tArrow{t}{\lin}{t}}$ (their unfolding makes it clear). This is because the same self-reference appears in both covariant and contravariant positions (hence the first must be simultaneously subtype and supertype of the latter, which cannot happen because of their multiplicities). We avoid generating such pairs in \cref{thm:props-rec} by ensuring that references introduced by a pair of recursive types only appear in covariant positions of their bodies, unless they are equivalent, in which case there is no such restriction.

To generate invalid subtyping pairs, we follow the same algorithm but inject also the invalid pairs that occur in each item of the following theorem. 
\begin{theorem}[Properties of subtyping (invalid)]
	\label{thm:props-subtyping-invalid}\
	\begin{enumerate}
		\item \label{item:invalid-1} $\tInt \not\subS \tUnit$, $\tUnit \not\subS \tInt$, $\tSkip \not\subS\tEnd$ and \item $\tEnd \not\subS \tSkip$;
		\item $\tRecord{\ell}{T_\ell}{L} \not\subS \tRecord{k}{U_k}{L}$ if $L \subsetneq K$ and $T_j\subS U_j (\forall j\in L)$;
		\item $\tVariant{\ell}{T_\ell}{L}\not\subS \tVariant{k}{U_k}{K}$ if $K \subsetneq L$ and $T_j \subS U_j (\forall j\in K)$; 
		\item $\tArrow{T}{\lin}{U} \not\subS \tArrow{V}{\un}{W}$, with $V\subS T, U\subS W$;
		\item\label{item:arrow1} $\tArrow{T}{m}{U} \not\subS \tArrow{V}{n}{W}$, with $T\subS V$, $T \not\sim V$, $m\sqsubseteq n$ and $U\subS W$;
		\item $\tArrow{T}{m}{U} \not \subS \tArrow{V}{n}{W}$, with $V\subS T$, $m\sqsubseteq n$, $W\subS U$ and $W\not\sim U$;
		\item $\tIn{T} \not\subS \tOut{T}$;
		\item $\tOut{T} \not\subS \tIn{T}$;
		\item $\tIn{T} \not\subS \tIn{U}$, with $T \not\sim U, U \subS T$;
		\item $\tOut{T} \not\subS \tOut{U}$, with $T\not\sim U, T \subS U$;
		\item $\tIntChoice{\ell}{T_\ell}{L} \not\subS \tIntChoice{k}{T_k}{K}$ if $L \subsetneq K$ and $T_j\subS U_j (\forall j\in L)$;
		\item $\tExtChoice{\ell}{T_\ell}{L}\not\subS \tExtChoice{k}{T_k}{K}$ if $K \subsetneq L$ and $T_j \subS U_j (\forall j\in K)$.
	\end{enumerate}
\end{theorem}
\begin{proof}
	By inspection of the syntactic subtyping rules, relying on the observation that if $T\subS U$ and $T\not\sim U$, then $U \not\subS T$.
\end{proof}
If $i = 1$, we generate one of the pairs in \cref{item:invalid-1} of \cref{thm:props-subtyping-invalid}, otherwise we use one of the items in \cref{thm:props-subtyping-valid} and randomly inject a subtyping pair where a valid one is supposed to be. If the generated pair turns out to be in the subtyping relation, we simply discard the result.

For example, suppose $i=4$. We randomly choose \cref{item:valid-5} of \cref{thm:props-subtyping-valid} to generate a pair of sequential compositions $(\tSeq{S_1}{R_1}, \tSeq{S_2}{R_2})$. We proceed in the valid path for the types before the semicolon, obtaining $S_1=\tSeq{\tOut{\tInt}}{\tSkip}$ and $S_2=\tOut{\tInt}$, but inject an invalid pair in $R_1$ and $R_2$. To generate it, we randomly choose \cref{item:arrow1} of \cref{thm:props-subtyping-valid}. Here we generate a valid pair $(T,U)$ using \cref{thm:props-subtyping-valid} and ensuring $T$ and $V$ are not equivalent, then multiplicities $m$ and $n$ such that $m \sqsubseteq n$, and finally another valid pair $(V,W)$. Thus we obtain $R_1=\tArrow{T}{m}{U}$ and $R_2=\tArrow{V}{n}{W}$, making $(\tSeq{S_1}{R_1}, \tSeq{S_2}{R_2})$ an invalid pair. If we had, however, generated $S_1=\tEnd$ and $S_2=\tSeq{\tEnd}{\tSkip}$, then $(\tSeq{S_1}{R_1}, \tSeq{S_2}{R_2})$ would be a valid pair, and its test result would be discarded.
\end{document}